\documentclass{article}

\usepackage[main, final]{neurips_2026}

\usepackage[utf8]{inputenc}
\usepackage[T1]{fontenc}
\usepackage{hyperref}
\usepackage{url}
\usepackage{booktabs}
\usepackage{amsfonts}
\usepackage{nicefrac}
\usepackage{microtype}
\usepackage{xcolor}

\usepackage{graphicx}
\usepackage{amsthm}
\usepackage{amsmath}
\usepackage{amssymb}
\usepackage{csquotes}
\usepackage{bbm}
\usepackage{dsfont}
\usepackage[boxruled,linesnumbered]{algorithm2e}
\usepackage{algpseudocode}
\usepackage[capitalise]{cleveref}

\newtheorem{theorem}{Theorem}
\newtheorem{corollary}[theorem]{Corollary}
\newtheorem{lemma}[]{Lemma}[section]

\newtheorem{definition}[]{Definition}
\newtheorem{remark}[]{Remark}

\newtheorem{innercustomthm}{Theorem}

\newenvironment{customthm}[1]
  {\renewcommand\theinnercustomthm{#1}\innercustomthm}
  {\endinnercustomthm}
\newtheorem{innercustomlemma}{Lemma}
\newenvironment{customlemma}[1]
  {\renewcommand\theinnercustomlemma{#1}\innercustomlemma}
  {\endinnercustomlemma}
  \newtheorem{innercustomcorollary}{Corollary}
\newenvironment{customcorollary}[1]
  {\renewcommand\theinnercustomcorollary{#1}\innercustomcorollary}
  {\endinnercustomcorollary}

\title{Domination-Avoiding Learning Agents Cannot Collude}

\author{
    Noam Nisan \\
    Hebrew University of Jerusalem \\
    \texttt{noam@cs.huji.ac.il}
    \And
    Emmanuel Zerah \\
    Hebrew University of Jerusalem \\
    \texttt{emmanuel.zerah@mail.huji.ac.il} 
}

\begin{document}

\maketitle

\begin{abstract}
An influential paper of Calvano et al. empirically demonstrated that Q-learning agents spontaneously collude when placed as sellers that compete on prices in a natural market model.  More recent results of Fish et al. empirically demonstrated that similar collusion happens with commercial LLMs.  
We formally prove that such collusion can also happen with external-regret-minimizing agents.  We identify a very general class of agents, which we term Domination-Avoiding agents, that provably do not collude in such markets. This class
contains all Mean-Based agents and all internal-regret-minimizing agents, as well as others such as Multiplicative-Weight agents with variable learning rate and contextual variants thereof.
More generally we show that, in any game, this class of agents is
guaranteed to jointly learn to almost never play strategies that are eliminated by repeated elimination of purely dominated strategies.
\end{abstract}

\section{Introduction}

In recent years we have seen more and more economic activity that is done by computerized
agents rather than directly by humans.  Prominent examples include auto-bidding in 
online ad auctions \cite{A+24} and high-speed trading in financial markets \cite{CJP15} but other examples abound, e.g. \cite{ACEX24,CMW16} and many more.  These algorithms are
designed to optimize outcomes for their ``owner'', learning and adapting to market conditions.  In some sense these learning agents are ``hyper-rational'', 
always following their internal algorithms, but they often deviate from standard
notions of either human or mathematical ``rationality'' due to various constraints such as learning limitations or
computational limitations.

When multiple such learning agents interact, the resulting emergent behavior often becomes
hard to analyze or predict.  Some authors have even used the term ``chaos'' to describe 
either the unpredictability \cite{AGPS26} of the dynamics of multiple agents or the various socially undesirable outcomes that it may have \cite{S+26}.  

It turns out that even rather simple learning agents operating in
rather simple economic settings exhibit somewhat ``mysterious'' behavior.  An influential
study \cite{CCDP20} looked at an economic setting 
of price competition in a market, specifically a Bertrand duopoly with ``logit demand''.
In this setting two producers need to each set a price for selling a product, where 
a lower price will result in getting a larger fraction of demand (rather than {\em all} the demand as in the classical Bertrand model).  If each producer 
optimizes for its own utility then the prices reach a well understood duopoly-price
Nash equilibrium.  If they collude then they can agree on the higher monopoly price and increase their revenue.  

In \cite{CCDP20}, they ran simulations of dynamics where 
the prices set by each producer were decided by standard classical ``Q-learning'' algorithms \cite{WD92, SB08} that are designed to optimize the agent's welfare
in an unknown dynamic environment.  The surprising empirical finding was that these Q-learning algorithms usually reached collusive supra-competitive prices, typically capturing roughly 70\% of the profit gap between the competitive and monopoly benchmarks.
While it is well known that long-term cooperation -- in our case, 
collusion -- can be maintained in repeated games (\cite{FT91, OR94}), this outcome
is quite surprising since these standard learning 
algorithms are designed to unilaterally optimize their own utility, do not communicate
with their ``competitor'' or try to influence their behavior using any (explicit) reward, 
punishment, or threat mechanisms.  More recent work has shown that
such collusion can also spontaneously develop among standard commercial LLM-based learning agents \cite{FGK24}.

So how do these agents learn to collude?  Will other learning agents collude as well?  How can we design learning agents 
that do not collude?
Can we design legal protections against
such algorithmic collusion \cite{HLZ24}?  On the flip side, even though in our context collusion -- cooperation between the agents -- is considered bad, in other contexts cooperation could be desirable and perhaps we could explicitly design {\em for} such cooperation.

Several recent theoretical papers have looked at these types of questions for theoretically clean classes of learning agents, specifically for various types of ``regret-minimizing'' agents (see textbook \cite{CL06}) in various economic contexts; see recent survey \cite{H26}. 

\begin{figure}[t]
    \centering
    \includegraphics[width=\textwidth]{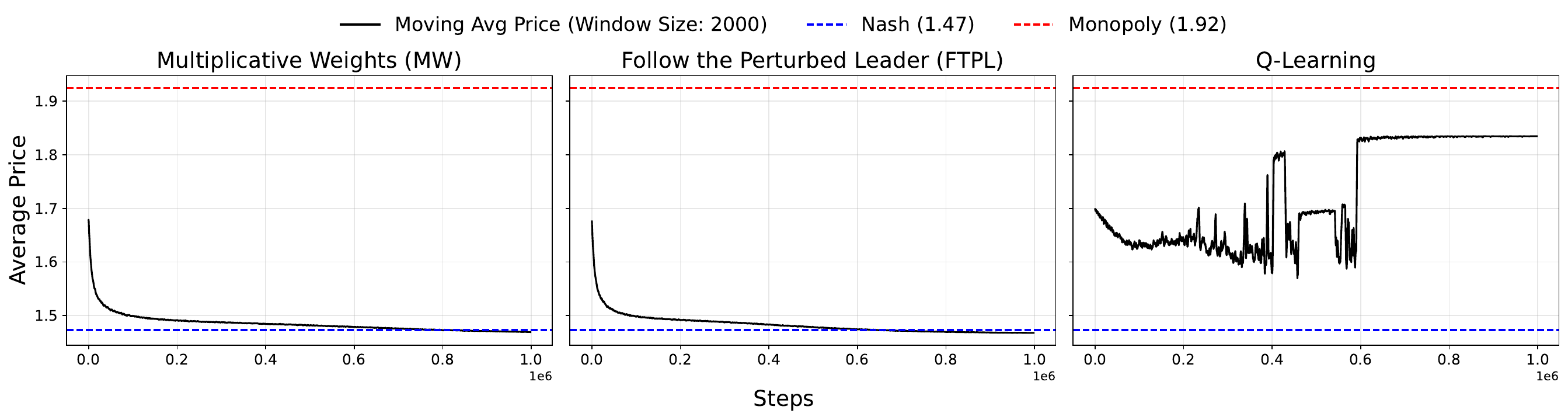}
    \caption{{Average pricing trajectories of learning agents in the Bertrand Logit duopoly game (typical run). Unlike Q-Learning (\textbf{Right}), which successfully coordinates on collusive prices, the no-regret algorithms MW (\textbf{Left}) and FTPL (\textbf{Center}) converge to the competitive Nash equilibrium.}}
    \label{fig:learning_dynamics}
\end{figure}

We started our investigation by running several ``classical''  external-regret-minimizing learning agents such as ``Multiplicative Weights'' (MW) and ``Follow the Perturbed Leader'' (FTPL) on the same Bertrand duopoly games with Logit demand considered by \cite{CCDP20,FGK24}. We observed that these algorithms did converge to the competitive equilibrium duopoly prices and did not ``collude'' -- in contrast to what was observed with Q-learning algorithms or LLMs {(see \cref{fig:learning_dynamics})}.  So the first natural question is {\em why?} Is this lack of collusion due to the external-regret minimization property itself or rather due to special properties of the specific algorithms that we tried? Our first result shows  external-regret-minimization by itself does not ensure
lack of collusion in this game. A similar possibility of collusion by external-regret minimizing algorithms was previously shown for classical Bertrand games (with low-price wins all market) \cite{NP10} as well as for auctions \cite{KN22, FLN16}, but we show it also for the technically more challenging Bertrand duopoly game with Logit demand as considered by \cite{CCDP20,FGK24}. In \cite{ACKRZ24} similar collusion was demonstrated in various cases where only one of the agents is regret-minimizing.

\begin{theorem} \label{thm:collusive_cce}
     For a simple symmetric Bertrand duopoly game with Logit demand 
     and no outside option, there exist external-regret-minimizing learning agents that in repeated play against each other reach 
     arbitrarily high collusive prices (even though the {Nash} 
     prices are finite).
\end{theorem}

Having established that external-regret minimization does not suffice to preclude collusion, 
one may naturally seek minimal conditions that do. Prior research highlights the role of Mean-Based learning \cite{BMS18} in driving competitive outcomes across various economic settings \cite{KN22, D+22, F+21, ACKRZ24}. The basic property of this class of learning algorithms is that they learn not
to play actions whose {\em average} performance so far has been inferior, and this class was shown by \cite{BMS18} to include several classical learning algorithms such as MW, FTPL, and EXP3.

Specifically, \cite{BDO24} establishes that Mean-Based learners succeed in iteratively eliminating dominated strategies in any finite game. As \cite{BDO24} also show, in classic Bertrand games this implies convergence to 
competitive prices, i.e. lack of collusion. (A conceptually related elimination dynamic has long been known in the different context of ``monotonic'' evolutionary strategies \cite{W97}). While we would not go as far as attempting to generally equate elimination of dominated strategies with lack of collusion, one may confidently say that playing a dominated strategy (even an iteratively dominated one) is a clear sign of either an irresponsible lack of optimization or some long term cooperative behavior -- in our case collusion.

We present two contributions to this literature.  Our technical contribution shows that in symmetric Bertrand games with Logit demand, iterative elimination of dominated strategies results in competitive prices. This extends the results of \cite{BDO24} from Bertrand games with standard (low-bid takes all) demand or ``linear demand'', to the more technically challenging Logit demand case considered by \cite{CCDP20,FGK24}. Our main conceptual contribution is extending the class of learning algorithms for which such lack of collusion is proven significantly beyond the class of Mean-Based algorithms.

\subsection*{Domination-Avoiding Agents}

Mean-Based learning algorithms must learn to avoid actions whose
mean performance so far has been inferior relative to some other action.  The definition of ``mean'' here is quite literal: the mathematical
average of the rewards achieved historically.  This is a rather
specific and fragile condition on the learning algorithm.  For example,
standard algorithms like Multiplicative-Weights cease to be Mean-Based 
if the learning rate is not constant or if they get some finite additional
context.  It is also known that no internal-regret minimizing algorithm can be Mean-Based \cite{DSS19}.  We wish to formalize a minimally-restrictive property on learning agents that will ensure lack of collusion.  Our definition will only require the agent to learn to avoid actions that have almost always been bad so far, not necessarily all those whose average performance was sub-par. This class will turn out to be significantly more general and robust than the class of Mean-Based agents.

Our setting is rather abstract and assumes that every time step $t=1,2, \dots,$ our agent plays some action $a^t \in A$ and receives some reward $r^t(a^t)$, where the reward function $r^t : A\rightarrow \mathbb{R}_+$ may be chosen, in general, by an adversary. In our application, it simply represents our agent's utility given the actual actions of the other agents in the game, i.e., $r^t_i(a_i)=u_i(a_i, a_{-i}^t)$. The natural model is
that the agent ``sees'' the full reward function $r^t(\cdot)$ (the ``experts'' setting) but the definitions apply to whatever informational setting we desire for the agent including the so-called bandit setting.  We first define 
what it means for an action $a$ to be ``empirically dominated'' at a certain time.

\begin{definition}[$\delta$-Empirical Domination]
    An action $a \in A$ is called $\delta$-empirically-dominated by $a' \in A$ up to time $t$ if:
    \[
    \left| \left\{ 1 \le s \le t \mid r^s(a) \ge r^s(a') \right\} \right| \le \delta t
    \]
    We say that $a$ is $\delta$-empirically-dominated (in short, $\delta$-dominated) up to time $t$ if
    for some $a' \in A$ it is $\delta$-empirically-dominated by $a'$ up to time $t$.
\end{definition}

Our requirement from a ``Domination-Avoiding'' agent is that it ``learns to almost never use empirically dominated actions''.  

\begin{definition} [Domination-Avoiding] \label{def:da}
    We say that a learning algorithm is Domination-Avoiding (DA) if for any adversary and any $\epsilon>0$
    there exists $\delta>0$ such that for all sufficiently large $T$ we have that 
    \[
    \Pr_{t \sim {U}(1 \dots T)}\left[a^t \text{ is } \delta\text{-dominated} \text{ up to time } t \right] \le \epsilon
    \]
    where the probability is taken over the uniform random choice of a time $t$ in the first $T$ time steps as well as over the internal randomization
    of the agent and the adversary.
\end{definition}

As we will observe momentarily, a variety of classes of learning agents satisfy
this Domination-Avoiding definition. After all, any learning algorithm worth its salt
should eventually learn to avoid playing actions that have consistently fared badly.
Our main positive result regarding Domination-Avoiding agents is that when such agents play a game against each other, then they collectively learn not to play
any action that is removable by iterated elimination of purely dominated strategies.

\begin{theorem}  \label{thm:da_IESPDS}
    Consider any finite $\mathcal{N}$-player game and any $\mathcal{N}$ Domination-Avoiding agents that repeatedly play this game against each other.
    For every action $a_i$ that is removable by Iterated Elimination of Strictly Purely Dominated Strategies we have that 
    \[
    \lim_{T\rightarrow \infty} \Pr_{t \sim {U}(1 \dots T)} [\text{player } i \text{ plays } a_i \text{ at time } t]=0
    \]
    where the probability is taken over a  {uniform} random choice of a time $t$ in the first $T$ time steps as well as the internal randomization of the agents.
\end{theorem}

We show that in a Bertrand duopoly game with Logit demand all strategies that are far from the competitive duopoly equilibrium prices are removable by Iterated Elimination of Strictly Purely Dominated Strategies. Formally we consider any finite pre-defined grid of possible prices {$p^{(1)}<p^{(2)}< \dots <p^{(M)}$, that is $\epsilon$-dense, $p^{(j)}-p^{(j-1)} \le \epsilon$, and whose range contains the equilibrium,  $p^{(1)} \le p^* \le p^{(M)}$,} and show that every price $p$ that is at distance more than $2\epsilon$ from the competitive Nash equilibrium price $p^*$ is removable by Iterated Elimination of Strictly Purely Dominated Strategies. From this we deduce our main corollary:

\begin{corollary}
    Consider two Domination-Avoiding agents that repeatedly play a symmetric Bertrand duopoly game with Logit demand over any finite grid of possible prices {$p^{(1)}<p^{(2)}< \dots <p^{(M)}$, that is $\epsilon$-dense, $p^{(j)}-p^{(j-1)} \le \epsilon$, and whose range contains the equilibrium,  $p^{(1)} \le p^* \le p^{(M)}$} where $p^*$ is the competitive duopoly equilibrium price (of the continuous Bertrand game). Then
    \[
    \lim_{T \rightarrow \infty} \Pr_{t \sim {U}(1 \dots T)}[ p^*-2\epsilon \le p_i^t \le p^*+2\epsilon] = 1
    \] 
    where $p_i^t$ is the price played by player $i$ at time $t$ and where the probability is taken over a uniformly random time $t$ in the first $T$ steps as well as over the internal randomization of the agents.
\end{corollary}

We now enumerate some of the types of learning algorithms that fall under this definition of Domination-Avoiding.

\paragraph{Fictitious Play} A fictitious play agent always best-responds to the distribution of actions of the other players that has been seen so far. Clearly any action that was almost always previously dominated cannot be such a best-reply.

\paragraph{Mean-Based Learners} Every Mean-Based learner is also a Domination-Avoiding learner.  This holds directly for the infinite-horizon variant of Mean-Based \cite{D+22}, while in the original formalism of Mean-Based that uses a sequence of finite horizons \cite{BMS18}, to get an infinite-horizon Domination-Avoiding learner, one uses the standard doubling strategy, i.e. run a sequence of finite horizon learning algorithms with horizon lengths $2, 4, 8, \dots$. As shown in \cite{BMS18}, a large number of well-known learning algorithms such as FTPL, EXP3, MW, are Mean-Based.

\paragraph{Multiplicative Weights} The Multiplicative Weights algorithm is Mean-Based, as long as the time-horizon is known and the learning rate $\eta$ is constant. Natural infinite-horizon implementations usually use a decreasing learning rate such as $\eta^t \propto 1/\sqrt{t}$, and they are not Mean-Based. We do show, however, that such variable-rate Multiplicative-Weights learning is still Domination-Avoiding.

\paragraph{Internal-Regret Minimizers} As opposed to general external-regret minimizing algorithms that may fail to be Domination-Avoiding, we show that internal-regret (or swap-regret) minimizers \cite{CL06} are always Domination-Avoiding.
It is interesting to note the contrast to the result of \cite{DSS19} that shows that {\em no} internal-regret minimizing learner can be Mean-Based.
    
\paragraph{Contextual Learning} Many learning algorithms act according to some ``extra'' information that they have access to. This may be some information that is gathered from the environment or may be some type of ``memory'' of the learning algorithm itself. Such information is often called ``context'' and a contextual variant of a learning algorithm will essentially run a separate copy of the underlying algorithm for each possible context. For instance, the S-EXP3 algorithm \cite{BC12} handles contextual bandits by running an independent copy of standard EXP3 for every possible context. As another example, the Q-learning algorithm, run by \cite{CCDP20}, that started our journey uses the context (state) of the last action of the player and its opponent. It turns out that the class of Domination-Avoiding learners is ``closed'' under finite contexts.  I.e., if there are finitely many possible contexts, and a Domination-Avoiding learning algorithm is run for each one of them separately, then the combined algorithm is Domination-Avoiding as well.

A significant part of our paper is proving that all these classes of learning agents are indeed Domination-Avoiding.

\begin{theorem}
    (informal) All the examples above are indeed Domination-Avoiding.
\end{theorem}

It is now perhaps a good time to go back to Q-learning and reflect on how it fails to be Domination-Avoiding, allowing such agents to collude \cite{CCDP20}. The first answer that one may wish to give is that Q-learning algorithms have ``state'' according to which they operate, specifically in the simulations of \cite{CCDP20} the prices offered depended on those in the previous time step.  But this by itself cannot be the solution to this mystery of collusion since such state falls under the category of ``context'' mentioned above which does leave us in the Domination-Avoiding class. The solution is in what Q-learning agents strive to optimize in every step: not just the immediate payoff but rather
the sum of the immediate payoff and the future expected payoffs from the new state that will be reached due to the action.  When such agents play the duopoly game against each other, each of them is able to learn that reducing the price in some step leads to a new state (that we may call ``price competition prone'') which leads to low payoff in the long term, and thus will avoid such an action. Such foresight is not Domination-Avoiding since a Domination-Avoiding agent cannot choose an action that is deemed to be bad
in the short term even if that action will later lead to long-term gains.  In a sense, Domination-Avoiding agents cannot pass the ``Marshmallow test'' \cite{ME70} of delaying gratification, as Q-learning agents are explicitly designed to do. We view this as the main limitation of Domination-Avoiding agents and the main reason why ``sufficiently clever'' agents will fail to meet this definition and be able to collude. We believe that analyzing and controlling collusion as well as beneficial cooperation between such ``sufficiently clever'' agents is an important direction for further research.

\section{Preliminaries and Notations}

\subsection{The Bertrand Logit Duopoly Model}

We consider the symmetric Bertrand pricing game between two competing producers. Each producer $i \in \{1, 2\}$ selects a price $p_i \in [0, \infty)$ for their product and captures a fraction of the total market demand, denoted as $Q_i$. To yield smoothed market shares -- unlike the abrupt demand shifts in the classic Bertrand model -- consumer demand is governed by a Logit function \cite{APT92} that includes an \emph{outside good} representing the consumer's option to not purchase. This specific formulation has recently been adopted as the standard benchmark environment in the algorithmic collusion literature \cite{CCDP20, FGK24, ACKRZ24}. In our symmetric setting, both producers offer products of identical ``quality'' (intrinsic utility to the buyers) $a$, while this outside option has quality $a_0 \in [-\infty, \infty)$.

The inclusion of the outside good means total market demand can be elastic; as prices rise, consumers can leave the market entirely. The fraction of demand captured by producer $i$, given the opponent's price $p_{-i}$, is:
\[
Q_i(p_i, p_{-i}) = \frac{e^{\mu(a - p_i)}}{e^{\mu(a - p_i)} + e^{\mu (a - p_{-i})} + e^{\mu a_0}}
\]
where $\mu > 0$ denotes consumer price sensitivity, governing the degree of product substitution. To simplify our notation for the relative strength of the outside option, we will use 
the constant $\alpha = e^{\mu (a_0 - a)}$. The boundary case where $a_0 = -\infty$ (yielding $\alpha = 0$) represents a market with no viable outside option, meaning total demand is perfectly inelastic and always sums to $1$.

Assuming a symmetric constant marginal cost ``of production'' $c \ge 0$, the expected payoff (profit) for player $i$ is:
\[
u_i(p_i, p_{-i}) = (p_i - c) \cdot Q_i(p_i, p_{-i})
\]
To evaluate whether algorithmic behavior is collusive, we must establish the standard economic benchmarks for this environment:

\paragraph{The Competitive Duopoly (Nash) Price:} The static Nash equilibrium is defined as the price profile where neither producer can strictly increase their expected profit by unilaterally deviating, meaning their prices are mutual best responses. In this continuous symmetric game, 
the standard analysis that solves the first-order conditions yields the unique competitive Nash equilibrium price $p^*$, which is implicitly defined by the equation:
\[
p^* = c + \frac{1}{\mu} \left[ 1 + \frac{1}{1 + \alpha e^{\mu p^*}} \right]
\]
In the boundary case with no outside option ($\alpha = 0$), this simplifies to the explicit price $p^* = c + \frac{2}{\mu}$.

\paragraph{The Monopoly Price:} Conversely, the monopoly price is the price that maximizes the producers' joint profit. In the boundary case with no viable outside option ($\alpha = 0$), total demand is perfectly inelastic, causing the optimal collusive price to diverge to infinity. When an active outside option exists ($\alpha > 0$), consumer exit bounds the joint profit, yielding a unique, finite monopoly price. This price serves as the theoretical ceiling for evaluating algorithmic collusion.

\paragraph{Finite Price Grid:} While the continuous game provides our theoretical economic benchmarks, standard algorithmic learning agents operate over finite action spaces. We formalize this by restricting the agents to choose prices from an arbitrary, finite, and bounded set of discrete prices $\mathcal{P} = \{p^{(1)}, p^{(2)}, \dots, p^{(M)}\}$ where $0 \le p^{(1)} < p^{(2)} < \dots < p^{(M)}$. We assume this price grid is \emph{$\epsilon$-dense}, meaning $p^{(j)} - p^{(j-1)} \le \epsilon$ for all $1 < j \le M$. We also assume $\mathcal{P}$ contains prices both below and above the competitive benchmark, $p^{(1)} \le p^* \le p^{(M)}$.

\paragraph{The Repeated Game Framework:} To analyze the emergent behavior of algorithmic pricing, we evaluate this static stage game within a repeated framework. A repeated game consists of an infinite sequence of discrete time steps $t \in \mathbb{N}$. At each step $t$, players simultaneously select a price from their finite grid $\mathcal{P}$ and receive the corresponding instantaneous payoff based on the realized joint price profile. The remainder of our analysis focuses on the asymptotic dynamics that unfold when learning agents interact in this repeated Bertrand game.

\subsection{Iterated Elimination of Strictly Purely Dominated Strategies}

To characterize the asymptotic behavior of learning agents, we rely on the standard solution concept of Iterated Elimination of Strictly Purely Dominated Strategies (IESPDS) \cite{FT91, OR94}. Consider a finite normal-form game $\Gamma = (\mathcal{N}, \{A_i\}_{i \in \mathcal{N}}, \{u_i\}_{i \in \mathcal{N}})$. A strategy is strictly dominated if there exists an alternative strategy that yields a strictly higher payoff against all valid opponent profiles. The IESPDS procedure recursively removes these dominated strategies.

\begin{definition}[Iterated Elimination of Strictly Purely Dominated Strategies]
    Let $\Gamma$ be a normal-form game. For each player $i$, let $S_i^0 = A_i$ be the initial set of available pure strategies. We recursively define the surviving set of strategies at round $k+1$ as:
    \[
    S_i^{k+1} = \left\{ a_i \in S_i^k \mid \nexists a'_i \in S_i^k \text{ such that } \forall a_{-i} \in S_{-i}^k, \ u_i(a'_i, a_{-i}) > u_i(a_i, a_{-i}) \right\}
    \]
    The Iterated Elimination of Strictly Purely Dominated Strategies (IESPDS) process repeatedly applies this condition. Because the game is finite, the sequence of sets weakly shrinks and must converge in finite steps to a stable surviving set $S_i^\infty = \bigcap_{k=0}^\infty S_i^k$.
\end{definition}

In subsequent sections, we will apply this iterative elimination process directly to the discrete Bertrand pricing grid $\mathcal{P}$ to bound the set of prices that can survive long-term interaction between DA learning agents.

\subsection{Domination-Avoiding -- Alternative Definition}

While the introductory definition of Domination-Avoiding (DA) agents is intuitive, mathematical derivations are simplified by an alternative formulation. The following equivalent definition is based on the mathematically convenient notion of Upper-Time-Average (UTA):

\begin{definition}[Upper-Time-Average]
    The upper-time-average (UTA) of an infinite sequence of real numbers $\{x^t\}_{t=1}^\infty$ is defined as:
    \[
    \operatorname{UTA}(\{x^t\}) = \limsup_{T \to \infty} \frac{1}{T} \sum_{t=1}^T x^t
    \]
\end{definition}

To formalize Domination-Avoiding agents using this notion, we define $d^t(\delta)$ as the probability -- evaluated over the joint randomization of the agent and the adversary -- that the agent selects a $\delta$-empirically dominated action at time $t$:
\[
d^t(\delta) = \Pr\left[ a^t \text{ is } \delta\text{-dominated}\text{ up to time } t \right]
\]

\begin{definition}[Domination-Avoiding Agent -- Alternative Definition] \label{def:alt_da}
    A learning agent is Domination-Avoiding (DA) if for any adversary and any $\epsilon > 0$, there exists $\delta > 0$ such that:
    \[
    \operatorname{UTA}(\{d^t(\delta)\}) \le \epsilon
    \]
\end{definition}

This asymptotic formulation is mathematically equivalent to the finite-time uniform sampling definition provided in the Introduction. The formal proof of this equivalence is deferred to \cref{app:sec2}.

\section{Regret-Minimizing Algorithms may be Collusive}

Regret minimizing algorithms \cite{CL06} are a widely studied class of algorithms that are often used as foundational mathematical models for how agents learn in unpredictable environments, so they are a natural starting point for our exploration of collusion by learning agents in the Bertrand pricing game studied by \cite{CCDP20, FGK24}.
We ran simulations of classic no-external-regret algorithms -- specifically the Multiplicative Weights (MW) and the Follow The Perturbed Leader (FTPL) algorithms.
We used the baseline parameterization of \cite{CCDP20}: product quality $a=2$, outside good quality $a_0=0$, marginal cost $c=1$, and price sensitivity $\mu=4$. The action space is restricted to a discrete grid of $15$ prices spanning the continuous competitive and monopoly benchmarks. Both no-regret algorithms utilize a learning rate of $5 \times 10^{-3}$, and all simulations are run for $10^6$ time steps to ensure robust limit behavior.

As illustrated in \cref{fig:learning_dynamics}, the empirical pricing trajectory reliably converges to the competitive Nash equilibrium and exhibits no evidence of the supra-competitive pricing typically observed with Q-learning agents. Given this robust empirical behavior, a natural hypothesis is that the no-external-regret property itself guarantees competitive outcomes in this environment.

We demonstrate this hypothesis is false. A well-known property of standard learning algorithms is that when multiple no-external-regret agents interact, their empirical joint frequency of play converges to the set of Coarse Correlated Equilibria (CCE). Conversely, for any CCE in a finite game, there exist sequences of play induced by no-external-regret learning dynamics that converge to that equilibrium \cite{CL06}.

\begin{definition}[Coarse Correlated Equilibrium] \label{def:CCE}
    Given a game $\Gamma$ with $\mathcal{N}$ players, with strategy space $A_i$ and utility function $u_i$ for each player $i$, a joint probability distribution $\sigma$ over $A=\prod_i A_i$ is a Coarse Correlated Equilibrium (CCE) if, for every player $i \in \mathcal{N}$ and every fixed unilateral deviation $a'_i \in A_i$, we have that
    \[
    \mathbb{E}_{(a_i,a_{-i}) \sim \sigma} [u_i(a_i,a_{-i})] \ge \mathbb{E}_{a_{-i} \sim \sigma_{-i}} [u_i(a'_i, a_{-i})]
    \]
\end{definition}

Because every CCE can be induced by no-external-regret learning dynamics, demonstrating that regret minimization permits algorithmic collusion reduces to exhibiting a single CCE supported on supra-competitive prices. We will exhibit this for the simplest case of a Bertrand market with Logit demand: no viable outside option ($\alpha = 0$), marginal cost $c = 0$ and price sensitivity $\mu = 1$. For these parameters the competitive Nash equilibrium is $p^* = 2$, while the theoretical monopoly price diverges to infinity.

{By establishing the existence of a collusive CCE in this normalized setting, the following result formally grounds the claim of \cref{thm:collusive_cce}}.

\begin{customthm}{{1$'$}} \label{thm:collusive_cce_formal}
    Consider the normalized Bertrand Logit game ($\alpha = 0$, $c = 0$, $\mu = 1$).
    For any target
    price {$V > p^*$, where $p^*$ is the competitive duopoly equilibrium price (of the continuous Bertrand game)}, there exists a Coarse Correlated Equilibrium (CCE) that is supported on
    strategies where both agents always bid 
    a price $p_i \ge V$.
\end{customthm}

To establish this result, we construct a symmetric CCE supported on exactly two discrete prices: a base price $p_L = V$ and an extreme bonus price $p_H \gg p_L$. Specifically, we define the joint strategy distribution $\sigma_{CCE}$ as:
\[
\sigma_{CCE} = \begin{cases}
(p_L, p_L) & \text{w.p.} \quad q \\
(p_H, p_H) & \text{w.p.} \quad 1-q
\end{cases}
\]
where the mixing probability $q = \frac{p_H}{p_L + p_H}$ exactly balances the expected payoff contributions from the two states. 

The formal proof is deferred to \cref{app:sec3}, where we precisely define $p_L$ and $p_H$ and prove that no unilateral deviation is profitable. This establishes $\sigma_{CCE}$ as a valid Coarse Correlated Equilibrium, demonstrating that standard no-external-regret algorithms can sustain arbitrarily high collusive prices.

\section{Convergence to Iteratively Undominated Strategies}

We analyze the interaction of $\mathcal{N}$ Domination-Avoiding (DA) agents through the infinite repetition of a finite stage game $\Gamma = (\mathcal{N}, \{A_i\}, \{u_i\})$. At each step $t \in \mathbb{N}$, players simultaneously select actions $a_i^t \in A_i$. We denote the marginal probability that player $i$ plays action $a_i$ at time $t$ as $p_i^t(a_i) = \Pr(a_i^t = a_i)$. This probability integrates both the agents' internal randomization and the evolving history of play.

To formalize the asymptotic behavior of these agents, we first define the set of actions that a learning agent continues to play with a strictly positive frequency over the infinite horizon.

\begin{definition}[Surviving Actions]
    An action $a_i \in A_i$ is said to survive if its upper-time-average probability of being played is strictly positive: $\operatorname{UTA}(\{p_i^t(a_i)\}) > 0$.
\end{definition}

Let $S_i \subseteq A_i$ denote the set of surviving actions for player $i$, and let $S_{-i} = \prod_{j \neq i} S_j$ denote the set of surviving joint profiles for the opponents of $i$.
To establish convergence, we first show that surviving actions must be justifiable. Specifically, the following lemma demonstrates that an action cannot survive if it is strictly dominated against the surviving actions of the opponents.

\begin{lemma}\label{lem:surviving_undominated_N}
    Suppose all $\mathcal{N}$ players use DA agents to repeatedly play the finite normal-form stage game $\Gamma$. For any player $i$, any surviving action $a_i \in S_i$, and every pure action $a'_i \in A_i$, there exists a surviving opponent profile $a_{-i} \in S_{-i}$ such that $u_i(a_i, a_{-i}) \ge u_i(a'_i, a_{-i})$.
\end{lemma}

The formal proof is deferred to \cref{app:sec4}. The argument proceeds by contradiction: if an alternative action $a'_i$ dominated $a_i$ against all surviving opponent profiles, $a_i$ could only yield a higher payoff when opponents play non-surviving actions. Because non-surviving actions are played with an asymptotic time-average of zero, $a_i$ becomes $\delta$-empirically dominated by its dominating strategy $a'_i$. A DA agent must therefore play $a_i$ with vanishing probability, contradicting the assumption that $a_i$ survives.

Lemma \ref{lem:surviving_undominated_N} establishes a strict local constraint: an action cannot survive unless it is justifiable against the surviving profiles of the opponents. By recursively applying this requirement, we arrive at our main global convergence result.

\begin{customthm}{2}
    Consider any finite $\mathcal{N}$-player game and any $\mathcal{N}$ Domination-Avoiding agents that repeatedly play this game against each other.
    For every action $a_i$ that is removable by Iterated Elimination of Strictly Purely Dominated Strategies we have that 
    \[
    \lim_{T\rightarrow \infty} \Pr_{t \sim {U}(1 \dots T)} [\text{player } i \text{ plays } a_i \text{ at time } t]=0
    \]
    where the probability is taken over a  {uniform} random choice of a time $t$ in the first $T$ time steps as well as the internal randomization of the agents.
\end{customthm}

The formal proof is deferred to \cref{app:sec4}. The argument proceeds by cascading elimination. DA agents asymptotically abandon empirically dominated actions. As players cease playing these initial actions (the first round of IESPDS), the effective game shrinks. Actions justifiable only against abandoned strategies become empirically dominated, causing agents to discard them next. This cascade recursively mirrors the IESPDS procedure, ensuring any iteratively eliminated strategy is played with asymptotic probability zero.

\paragraph{Application to Bertrand Logit Pricing}
To apply this general convergence guarantee to our economic setting, we must characterize the IESPDS set of our discrete Logit pricing game over the $\epsilon$-dense grid $\mathcal{P}$. As the following theorem establishes, this iterative elimination process systematically removes non-competitive regions of the action space from both the top down and the bottom up.

\begin{theorem} \label{thm:IESPDS_collapse}
    Consider a symmetric Bertrand Logit pricing game over a finite $\epsilon$-dense grid $\mathcal{P}$ of possible price levels. Let $p^*$ be the Nash equilibrium of the continuous game, and assume $\min(\mathcal{P}) \le p^* \le \max(\mathcal{P})$.
    Then only prices $p \in \mathcal{P}$ satisfying $|p-p^*| \le 2 \epsilon$ can survive Iterated Elimination of Strictly Purely Dominated Strategies (IESPDS).
\end{theorem}

The full proof appears in \cref{app:sec4}. Combining this structural characterization with the asymptotic behavioral guarantee of Theorem \ref{thm:da_IESPDS} yields our main economic result: Domination-Avoiding agents converge to competitive pricing and fail to sustain collusion.

\begin{customcorollary}{3}
    Consider two Domination-Avoiding agents that repeatedly play a symmetric Bertrand duopoly game with Logit demand over any finite grid of possible prices {$p^{(1)}<p^{(2)}< \dots <p^{(M)}$, that is $\epsilon$-dense, $p^{(j)}-p^{(j-1)} \le \epsilon$, and whose range contains the equilibrium,  $p^{(1)} \le p^* \le p^{(M)}$} where $p^*$ is the competitive duopoly equilibrium price (of the continuous Bertrand game). Then
    \[
    \lim_{T \rightarrow \infty} \Pr_{t \sim {U}(1 \dots T)}[ p^*-2\epsilon \le p_i^t \le p^*+2\epsilon] = 1
    \] 
    where $p_i^t$ is the price played by player $i$ at time $t$ and where the probability is taken over a uniformly random time $t$ in the first $T$ steps as well as over the internal randomization of the agents.
\end{customcorollary}

The formal proof is deferred to \cref{app:sec4}. The idea is as follows: Recall from \cref{thm:IESPDS_collapse} that applying IESPDS to the discrete Bertrand Logit game eliminates all prices outside a $2\epsilon$-neighborhood of the competitive Nash equilibrium $p^*$. By \cref{thm:da_IESPDS}, DA agents asymptotically cease playing any strategy eliminated by IESPDS. Because the discrete grid is finite, the agents' probability mass must concentrate entirely on the surviving competitive window.

\section{Standard Algorithms are Domination-Avoiding}

The Domination-Avoiding (DA) framework encompasses a broad spectrum of standard learning dynamics in normal-form games. This class includes all Mean-Based algorithms across both finite (using the Doubling Trick) and infinite horizons. It also captures the classic infinite-horizon variant of Multiplicative Weights operating under standard decaying learning rates, which is not Mean-Based. Furthermore, the DA property establishes a theoretical separation in the collusive potential of regret minimizers: while no-external-regret guarantees permit supra-competitive pricing (\cref{thm:collusive_cce}), any algorithm guaranteeing no-internal-regret satisfies the DA condition. Finally, the DA property is closed under finite contexts. Consequently, agents that condition their pricing on market histories are also DA. The formal definitions of these algorithmic families, together with the complete proofs establishing their Domination-Avoiding property, are deferred to \cref{app:sec5}.

\section{Conclusion and Future Work}

In this paper, we demonstrated that not all learning algorithms are prone to spontaneous algorithmic collusion. By formalizing the class of Domination-Avoiding (DA) agents, we proved that algorithms which consistently avoid empirically dominated strategies fundamentally fail to sustain collusive outcomes in settings such as a Bertrand duopoly. Driven by the iterative elimination of strictly dominated high-price actions, DA agents inevitably converge to the competitive Nash equilibrium.

While our findings characterize the conditions under which independent learning agents \emph{fail} to collude, this naturally raises an important inverse question for future research: what specific algorithmic properties and environmental conditions are required to \emph{enable} and sustain collusion? In the context of economic pricing, algorithmic collusion is heavily scrutinized as it harms consumer welfare. However, in broader multi-agent settings -- such as autonomous navigation, network routing, and decentralized resource allocation -- the ability for independent agents to spontaneously cooperate is highly desirable. Future work will aim to identify the mechanisms necessary for algorithms to learn cooperative policies, helping us design agents that can successfully coordinate when beneficial, while understanding how to prevent them from doing so when it is harmful.

\begin{ack}
This research was supported by a grant from the Israeli Science Foundation (ISF number 505/23).
\end{ack}

\bibliographystyle{unsrt}
\bibliography{bib}

\appendix

\crefalias{section}{appendix}

\section{Classes of Domination-Avoiding Algorithms}
\label{app:sec5}

We establish the DA property for four broad classes of algorithms: Mean-Based dynamics (including finite-horizon variants via the Doubling Trick), the standard infinite-horizon Multiplicative Weights algorithm (with decaying learning rate), No-Swap-Regret agents, and Contextual Domination-Avoiding algorithms.
Throughout this section, we analyze these dynamics from the perspective of a single agent facing an arbitrary environment. Dropping the player subscript, let $A$ denote the agent's finite action space and $r^t(a) = u(a, a_{-i}^t)$ denote the instantaneous reward of playing $a$ at time $t$. For notational convenience, we normalize the reward space to $r^t(a) \in [0,1]$, noting that all bounds scale proportionally for any bounded utility interval, and denote by $c=\min_{a,a' \in A, t, \text{ such that } 
r^t(a)\ne r^t(a')} |r^t(a)-r^t(a')|$ the minimum gap of non-equal utilities in the game.  Since the game is finite we have that $c>0$ is defined.  
More abstractly, when just talking about a single agent playing against an adversary, when such $c>0$ exists we say that the environment satisfies the minimum reward gap condition.

Throughout this appendix, we analyze the learning dynamics using the Upper-Time-Average (UTA) formulation of Domination-Avoiding agents (\cref{def:alt_da}), as its structure simplifies the technical mathematics.

\subsection{Mean-Based Learners}

We begin with Mean-Based (MB) learning algorithms. A MB agent is characterized by a strict behavioral constraint: it selects an action with vanishing probability if its cumulative utility lags the optimal in hindsight action. 

This framework is formalized under two standard paradigms. The original formulation \cite{BMS18} assumes a known, finite time horizon, while subsequent work \cite{F+21, D+22, BDO24} extended this to infinite-horizon dynamics. We establish that both variants satisfy the Domination-Avoiding property.

This connection resolves the DA status of several canonical learning algorithms. In Fictitious Play, an agent best-responds to the historical empirical distribution of its opponents; because the empirical utility of a historically dominated action trails the optimal response, it satisfies the infinite-horizon MB definition. Conversely, standard no-regret algorithms including Follow the Perturbed Leader (FTPL), EXP3, and Multiplicative Weights satisfy the finite-horizon MB definition \cite{BMS18}. As we show below, extending these to the infinite horizon via the standard Doubling Trick preserves the Domination-Avoiding guarantee.

\paragraph{Infinite-Horizon Mean-Based Algorithms}

The infinite-horizon formulation evaluates behavior as time progresses to infinity, serving as a direct conceptual analogue to our Domination-Avoiding framework. The definition penalizes historically suboptimal actions: if an action trails the empirical leader by a linearly growing margin, its selection probability must decay to zero.

\begin{definition}[Mean-Based Learning -- Infinite Horizon \cite{F+21, D+22, BDO24}]
    Let $\sigma^t(a) = \sum_{s=1}^t r^s(a)$ denote the cumulative reward of action $a \in A$ up to round $t$. A learning algorithm is Mean-Based in the infinite-horizon regime if there exists a monotonically decreasing function $\gamma: \mathbb{N} \to \mathbb{R}^+$ with $\gamma(t) = o(1)$, such that for all $t$ and any pair of actions $a, a' \in A$, whenever:
    \[
    \sigma^t(a') - \sigma^t(a) > \gamma(t) \cdot t
    \]
    the probability of playing action $a$ at time $t$ satisfies $\Pr(a^t = a) \le \gamma(t)$.
\end{definition}

\begin{customthm}{4.1} \label{thm:mb_infinite}
    Any infinite-horizon Mean-Based learning algorithm is Domination-Avoiding (DA) against any environment satisfying the minimum reward gap condition.
\end{customthm}

\begin{proof}
Fix an arbitrary $\epsilon > 0$. We will show there exists $\delta > 0$ such that $\operatorname{UTA}(\{d^t(\delta)\}) = 0 < \epsilon$.

Suppose an action $a \in A$ is $\delta$-empirically dominated by $a' \in A$ up to time $t$. We bound the cumulative reward gap 
\[
\Delta^t = \sigma^t(a') - \sigma^t(a)=\sum_{s=1}^t \left( r^s(a')-r^s(a) \right)
\] 
By the definition of $\delta$-domination, $r^s(a) \ge r^s(a')$ in at most $\delta t$ rounds. For the remaining $(1-\delta)t$ rounds, the strict dominance $r^s(a') > r^s(a)$ holds.

Because rewards are normalized to $[0, 1]$, the instantaneous reward difference $r^s(a') - r^s(a)$ in the exception rounds is bounded below by $-1$. By the minimum reward gap condition, this difference in the remaining rounds is at least $c > 0$. Bounding the cumulative gap yields:
\[
\Delta^t \ge c(1-\delta)t - \delta t = t \big(c - \delta(c+1)\big)
\]
Setting $\delta = \frac{c}{2(c+1)}$ gives:
\[
\Delta^t \ge \frac{c}{2} t
\]

By the definition of an infinite-horizon Mean-Based algorithm, the probability of playing $a$ is bounded by $\gamma(t)$ whenever $\Delta^t > \gamma(t) \cdot t$. Substituting our lower bound, this condition is satisfied whenever $\frac{c}{2} > \gamma(t)$. 

Since $\gamma(t) = o(1)$, there exists a finite time $T^*$ such that $\gamma(t) < \frac{c}{2}$ for all $t > T^*$. Thus, for $t > T^*$, the marginal probability of playing any specific $\delta$-dominated action is at most $\gamma(t)$. 

Recall that $d^t(\delta)$ denotes the overall probability that the agent plays any $\delta$-dominated action at time $t$. Applying a union bound over the finite action space $A$, we obtain $d^t(\delta) \le |A|\gamma(t)$ for all $t > T^*$. For early rounds ($t \le T^*$), we use the trivial bound $d^t(\delta) \le 1$.

We evaluate the Upper-Time-Average (UTA):
\[
\operatorname{UTA}(\{d^t(\delta)\}) = \limsup_{T \to \infty} \frac{1}{T} \sum_{t=1}^T d^t(\delta) \le \limsup_{T \to \infty} \frac{1}{T} \left( T^* + |A| \sum_{t=T^*+1}^{T} \gamma(t) \right)
\]
Because $T^*$ is finite, the first term vanishes. Since $\lim_{t \to \infty} \gamma(t) = 0$, Cauchy's Limit Theorem ensures that the sequence of its arithmetic means also converges to $0$. As $|A|$ is constant, the limit superior evaluates to $0$. This satisfies the Domination-Avoiding condition.
\end{proof}

\paragraph{Finite-Horizon Mean-Based Algorithms and the Doubling Trick}

Many standard learning algorithms are analyzed over a fixed, finite time horizon $T$. To deploy such algorithms in an infinite-horizon setting without knowing $T$ a priori, a standard technique is the ``Doubling Trick.''

\begin{definition}[Mean-Based Learning -- Finite Horizon \cite{BMS18}]
    Let $\mathcal{A} = \{\mathcal{A}_T\}_{T=1}^\infty$ be a family of learning algorithms parameterized by a known time horizon $T$. Let $\sigma^t(a) = \sum_{s=1}^t r^s(a)$ denote the cumulative reward of action $a \in A$ up through round $t \le T$. The family $\mathcal{A}$ is Mean-Based if there exists a function $\gamma: \mathbb{N} \to \mathbb{R}^+$ with $\gamma(T) = o(1)$ as $T \to \infty$, such that for any horizon $T$, any time step $t \le T$, and any pair of actions $a, a' \in A$, whenever:
    \[
    \sigma^t(a') - \sigma^t(a) > \gamma(T) \cdot T
    \]
    the probability of algorithm $\mathcal{A}_T$ playing action $a$ at time $t$ satisfies $\Pr(a^{t} = a) \le \gamma(T)$.
\end{definition}

\begin{customthm}{4.2} \label{thm:mb_finite_doubling}
    Let $\mathcal{A}$ be a finite-horizon Mean-Based learning algorithm. Consider an agent that runs $\mathcal{A}$ in consecutive epochs $n = 0, 1, 2, \dots$ of length $T_n = 2^n T_0$ (for some base constant $T_0 > 0$), completely resetting its internal memory at the beginning of each epoch. This agent is Domination-Avoiding against any environment satisfying the minimum reward gap condition.
\end{customthm}

\begin{proof}
Fix an arbitrary $\epsilon > 0$. We will show there exists $\delta > 0$ such that $\operatorname{UTA}(\{d^t(\delta)\}) \le \epsilon$.

The algorithm runs in epochs of length $T_n = 2^n T_0$. The absolute time $t_n$ at the start of epoch $n$ is $\sum_{k=0}^{n-1} 2^k T_0 = (2^n - 1)T_0$. Notice that $t_n < T_n$. Let $t$ be the current absolute time step, falling within epoch $n$ at relative time $\tau \in [1, T_n]$. Thus, $t = t_n + \tau < 2T_n$.

Suppose that up to time $t$, an action $a \in A$ is globally $\delta$-empirically dominated by $a' \in A$. By definition, the total number of exception rounds where $r^s(a) \ge r^s(a')$ is at most $\delta t$. Because $t < 2T_n$, the absolute number of exceptions up to $t$ is strictly bounded by $2\delta T_n$.

Because the internal memory resets at the beginning of epoch $n$, the algorithm evaluates the cumulative gap accumulated exclusively from the epoch start:
\[
\Delta_n^\tau = \sum_{s=t_n+1}^{t_n+\tau} \big( r^s(a') - r^s(a) \big)
\]
In the worst case, all $2\delta T_n$ global exceptions occur within this active memory. 

We partition these $\tau$ rounds based on the instantaneous difference $r^s(a') - r^s(a)$. In the exception rounds ($r^s(a) \ge r^s(a')$), the normalized $[0,1]$ rewards bound this difference below by $-1$. In the remaining dominating rounds ($r^s(a') > r^s(a)$), the minimum reward gap condition bounds this difference below by $c > 0$. Thus, the internal cumulative gap at relative time $\tau$ satisfies:
\[
\Delta_n^\tau \ge c(\tau - 2\delta T_n) - (1)(2\delta T_n) = c\tau - 2\delta(c+1)T_n
\]

By the definition of a finite-horizon MB algorithm, the probability of playing action $a$ is bounded by $\gamma(T_n)$ whenever $\Delta_n^\tau > \gamma(T_n) \cdot T_n$. Substituting our lower bound, this condition triggers when:
\[
\tau > T_n \left( \frac{2\delta(c+1) + \gamma(T_n)}{c} \right)
\]
Let $\alpha_n = \frac{2\delta(c+1) + \gamma(T_n)}{c}$ denote the maximum fraction of epoch $n$ spent in this initial unpenalized phase for a single action. During this phase ($\tau \le \alpha_n T_n$), the marginal probability of playing the dominated action is bounded by $1$. Afterward ($\tau > \alpha_n T_n$), it is bounded by $\gamma(T_n)$. Applying a union bound over all actions in $A$, the expected marginal error $E_n$ (the probability of playing any dominated action) averaged over epoch $n$ satisfies:
\[
E_n \le |A|\alpha_n + |A|\gamma(T_n)
\]
Since $\lim_{n \to \infty} \gamma(T_n) = 0$, taking the limit superior of the epoch error bound yields:
\[
\limsup_{n \to \infty} E_n \le |A| \lim_{n \to \infty} \alpha_n = \frac{2\delta|A|(c+1)}{c}
\]
We choose $\delta = \frac{c\epsilon}{8|A|(c+1)} > 0$, which guarantees $\limsup_{n \to \infty} E_n \le \frac{\epsilon}{4}$.

Finally, we evaluate the Upper-Time-Average over absolute time $t \to \infty$. For $t$ inside epoch $N$, the time-averaged error is bounded by the total accumulated errors over all past epochs plus the worst-case accumulation in the current epoch. Since $t \ge t_N$, we can bound the denominator:
\[
\frac{1}{t} \sum_{s=1}^{t} d^s(\delta) \le \frac{\sum_{n=0}^{N-1} E_n T_n + E_N T_N}{t_N} = \frac{\sum_{n=0}^{N-1} E_n T_n}{\sum_{n=0}^{N-1} T_n} + E_N \frac{T_N}{t_N}
\]
We evaluate the limit superior as $N \to \infty$. The first fraction is a weighted average of the sequence $E_n$; because the total elapsed time diverges ($\sum T_n \to \infty$), the limit superior of this average is bounded by the limit superior of the sequence itself ($\limsup E_n \le \frac{\epsilon}{4}$). 

For the second term, the ratio $\frac{T_N}{t_N} = \frac{2^N}{2^N - 1}$ converges to $1$, yielding an asymptotic bound of $\limsup E_N \cdot 1 \le \frac{\epsilon}{4}$. Thus:
\[
\operatorname{UTA}(\{d^t(\delta)\}) \le \frac{\epsilon}{4} + \frac{\epsilon}{4} = \frac{\epsilon}{2} < \epsilon
\]
This satisfies the Domination-Avoiding condition.
\end{proof}

\subsection{Multiplicative Weights with Variable Learning Rate}

The Multiplicative Weights (MW) algorithm (\cref{alg:mw}) is a cornerstone of online learning. While MW with a constant learning rate is Mean-Based and regret-minimizing in finite-horizon settings, it fails to achieve vanishing regret over an infinite horizon. To deploy MW over an unknown, infinite horizon, standard implementations employ a decaying learning rate (typically $\eta^t \propto 1/\sqrt{t}$) to secure optimal no-regret bounds. Because this decaying-rate variant violates the infinite-horizon Mean-Based condition, its status as a Domination-Avoiding agent is not immediate. We establish that it remains Domination-Avoiding.

\IncMargin{1em}
\begin{algorithm}[ht]
\caption{Multiplicative Weights}
\label{alg:mw}
\SetKwInOut{Require}{Require}
\Require{Finite action space $A$, decaying learning rate schedule $\{\eta^t\}_{t \ge 1}$ with $0 < \eta^t < 1$}
\textbf{Initialize:} $w^1(a) = 1$ for all actions $a \in A$\;
\For{$t = 1, 2, \dots$}{
    Choose action $a \in A$ with probability $p^t(a) = w^t(a) \big/ \sum_{a' \in A} w^t(a')$\;
    Observe the reward $r^t(a) \in [0, 1]$ for all $a \in A$\;
    Update the weight for all $a \in A$:
    \[
        w^{t+1}(a) = w^t(a) \cdot \exp\left(\eta^t r^t(a)\right)
    \]
}
\end{algorithm}
\DecMargin{1em}

\begin{customthm}{4.3} \label{thm:mw_decaying}
    The Multiplicative Weights algorithm with a decaying learning rate $\eta^t = \lambda{t^{-\beta}}$, for any constants $\lambda > 0$ and $\beta \in (0, 1)$, is Domination-Avoiding against any adversary satisfying the minimum reward gap condition.
\end{customthm}

\begin{proof}
Fix an arbitrary $\epsilon > 0$. We will show there exists $\delta > 0$ such that $\operatorname{UTA}(\{d^t(\delta)\}) = 0 < \epsilon$.

Fix an arbitrary time horizon $T > 1$. Suppose that up to time $T$, an action $a \in A$ is $\delta$-empirically dominated by an alternative action $a' \in A$. Unfolding the recursive definition of the algorithm, the probability of playing action $a$ at time $T$ is bounded by its weight ratio relative to the dominant action:
\[
p^T(a) = \frac{w^T(a)}{\sum_{b \in A} w^T(b)} \le \frac{w^T(a)}{w^T(a')} = \exp\left(-\sum_{t=1}^{T-1} \eta^t \big(r^t(a') - r^t(a)\big)\right)
\]

Let $\Delta^t = r^t(a') - r^t(a)$. We partition the history $\{1, \dots, T-1\}$ into exception rounds ($r^t(a) \ge r^t(a')$) and dominating rounds ($r^t(a') > r^t(a)$). By the definition of $\delta$-domination, there are at most $\delta T$ exception rounds, yielding a worst-case gap of $\Delta^t \ge -1$ under $[0,1]$ normalized rewards. For the remaining dominating rounds, the minimum reward gap condition guarantees $\Delta^t \ge c$.

We lower-bound the cumulative exponent sum by attributing at least $c$ to every round, and subtracting the maximum possible penalty $(c+1)$ for the exception rounds:
\[
\sum_{t=1}^{T-1} \eta^t \Delta^t \ge c \sum_{t=1}^{T-1} \eta^t - (c+1) \sum_{t \in \text{exceptions}} \eta^t
\]

Substituting the learning rate $\eta^t = \lambda{t^{-\beta}}$, we establish the worst-case lower bound by maximizing this subtracted term. Because the learning rate is strictly decreasing, the subtracted term is maximized if all exception rounds occur immediately at the start of learning ($t=1, \dots, \lfloor \delta T \rfloor$):
\[
\sum_{t=1}^{T-1} \eta^t \Delta^t \ge \lambda \left( c \sum_{t=1}^{T-1} t^{-\beta} - (c+1) \sum_{t=1}^{\lfloor \delta T \rfloor} t^{-\beta} \right)
\]

Because the function $x^{-\beta}$ is strictly decreasing for $\beta \in (0, 1)$, we bound the term $t^{-\beta}$ by its integral over adjacent unit intervals: 
\[
\int_{t}^{t+1} x^{-\beta} \, dx < t^{-\beta} < \int_{t-1}^{t} x^{-\beta} \, dx
\]
Evaluating these integrals yields
\[
\frac{1}{1-\beta}\big((t+1)^{1-\beta} - t^{1-\beta}\big) < t^{-\beta} < \frac{1}{1-\beta}\big(t^{1-\beta} - (t-1)^{1-\beta}\big)
\]
Summing these strict bounds telescopes directly:
\begin{align*}
\sum_{t=1}^{T-1} t^{-\beta} &> \sum_{t=1}^{T-1} \frac{1}{1-\beta} \big((t+1)^{1-\beta} - t^{1-\beta}\big) = \frac{T^{1-\beta} - 1}{1-\beta} \\
\sum_{t=1}^{\lfloor \delta T \rfloor} t^{-\beta} &< \sum_{t=1}^{\lfloor \delta T \rfloor} \frac{1}{1-\beta} \big(t^{1-\beta} - (t-1)^{1-\beta}\big) = \frac{\lfloor \delta T \rfloor^{1-\beta}}{1-\beta} \le \frac{(\delta T)^{1-\beta}}{1-\beta}
\end{align*}
Substituting these bounds back into our inequality yields:
\begin{align*}
\sum_{t=1}^{T-1} \eta^t \Delta^t &> \lambda \left[ c \left( \frac{T^{1-\beta} - 1}{1-\beta} \right) - (c+1) \frac{\delta^{1-\beta} T^{1-\beta}}{1-\beta} \right] \\
&= \frac{\lambda}{1-\beta} T^{1-\beta} \big( c - (c+1)\delta^{1-\beta} \big) - \frac{\lambda c}{1-\beta}
\end{align*}

We select $\delta < \left(\frac{c}{c+1}\right)^{\frac{1}{1-\beta}}$, ensuring the term $c - (c+1)\delta^{1-\beta}$ is strictly positive.

Let $K_1 = \frac{\lambda}{1-\beta} \big(c - (c+1)\delta^{1-\beta}\big) > 0$ denote this fixed decay rate. Substituting the bounded sum back into our initial probability ratio, we separate the static leading constant from the time-decaying exponent:
\[
p^T(a) \le \exp\left( \frac{\lambda c}{1-\beta} \right) \exp\Big(-K_1 T^{1-\beta}\Big)
\]

Let $K_2 = \exp\left( \frac{\lambda c}{1-\beta} \right)$. Applying a union bound over the finite action space $A$, the overall marginal probability $d^T(\delta)$ of the agent playing any $\delta$-dominated action at time $T$ is bounded by:
\[
d^T(\delta) \le |A| K_2 \exp\Big(-K_1 T^{1-\beta}\Big)
\]

Finally, we evaluate the Upper-Time-Average. Because $\beta \in (0, 1)$, the exponent $1-\beta$ is positive, so $T^{1-\beta} \to \infty$ as $T \to \infty$. Because the decay rate $K_1$ is strictly positive, the exponential term decays to $0$. Consequently, the sequence $d^T(\delta)$ vanishes. Because the sequence converges to $0$, the limit superior of its time-average evaluates to $0$. Thus, $\operatorname{UTA}(\{d^t(\delta)\}) = 0 < \epsilon$, satisfying the Domination-Avoiding condition.
\end{proof}

\subsection{No-Swap Regret Minimizers}

\cref{thm:collusive_cce} established that no-external-regret agents can sustain arbitrarily high collusive prices in the Logit duopoly game. A strictly stronger theoretical benchmark is No-Swap-Regret (NSR), also known as internal-regret, which guarantees an agent does not regret playing any specific action over another in hindsight. It is well known that swap-regret-minimizers that play a game against each other converge to the set of correlated equilibria of the game -- a subset of the coarse correlated equilibria \cite{CL06}. We prove that while external regret minimizers may fail the Domination-Avoiding condition, any algorithm guaranteeing No-Swap-Regret strictly satisfies it. 

\begin{definition}[No-Swap-Regret \cite{CL06}]
    At each time step $t \in \mathbb{N}$, the agent selects an action $a^t \in A$ and an adversary (or environment) reveals a reward function $r^t: A \to [0, 1]$. For any swap function $\phi: A \to A$, the swap regret at horizon $T$ is:
    \[
    R^\phi(T) = \sum_{t=1}^T \left( r^t(\phi(a^t)) - r^t(a^t) \right)
    \]
    The total swap regret is $R_{\text{swap}}(T) = \max_{\phi:A \rightarrow A} R^\phi(T)$. An algorithm guarantees No-Swap-Regret if, against any adversary, its expected time-averaged swap regret asymptotically converges to zero (or is non-positive):
    \[
    \limsup_{T \to \infty} \frac{1}{T} \mathbb{E}\left[ R_{\text{swap}}(T) \right] \le 0
    \]
\end{definition}

\begin{customthm}{4.4} \label{thm:nsr_da}
    Any learning algorithm that guarantees No-Swap-Regret is Domination-Avoiding against any environment satisfying the minimum reward gap condition.
\end{customthm}

\begin{proof}
We prove the contrapositive: if a learning agent is not Domination-Avoiding (DA), it cannot guarantee No-Swap-Regret (NSR).

Assume the agent is not DA. Then there exists an adversary and $\epsilon > 0$ such that for all $\delta > 0$, $\operatorname{UTA}(\{d^t(\delta)\}) > \epsilon$. 
Define $\epsilon' = \frac{\epsilon}{|A|^2}$. We fix a specific threshold $\delta = \frac{c\epsilon'}{2(c+1)}$. Because the DA violation holds for all $\delta > 0$, it holds for this chosen $\delta$.

Let $D^t$ be the event that action $a$ is $\delta$-empirically dominated by $a'$ up to time $t$. The limit superior of a finite sum is bounded by the sum of the limits superior. Applying a union bound over the $|A| \cdot (|A| - 1)$ distinct action pairs guarantees the existence of a fixed pair $(a, a')$ and an infinite sequence of horizons $\{T_k\}_{k=1}^\infty$ such that:
\[
\frac{1}{T_k} \sum_{t=1}^{T_k} \Pr \left(a^t = a \cap D^t\right) > \epsilon'
\]

Define the swap function:
\[
\phi(x) = \begin{cases}
a' & \quad x=a \\
x & \quad \text{otherwise}
\end{cases}
\]
The regret for this swap $\phi$ at time $t$ is:
\[
R^{\phi}_t = \sum_{s=1}^t \mathbbm{1}\{a^s = a\} \left( r^s(a') - r^s(a) \right)
\]
Define the random variable $\tau_k = \max \{t \le T_k \mid D^t \text{ occurs}\}$, with $\tau_k = 0$ if the event never occurs. Let the random variable $N_k$ denote the total number of times the agent plays $a$ while $D^t$ holds. By definition, $D^t$ does not hold for any $t > \tau_k$, meaning $N_k$ accumulates entirely within the first $\tau_k$ rounds. Consequently, the agent plays $a$ at least $N_k$ times up to round $\tau_k$.

To lower-bound the realized regret at $\tau_k$, we partition the rounds up to $\tau_k$ where the agent plays $a$ into exception rounds ($r^s(a) \ge r^s(a')$) and dominating rounds ($r^s(a') > r^s(a)$). Since $D^{\tau_k}$ holds, by definition of $\delta$-domination, there are at most $\delta \tau_k \le \delta T_k$ exception rounds, each contributing at worst $-1$ to the regret. Because the agent plays $a$ at least $N_k$ times, and at most $\delta T_k$ of those can be exception rounds, at least $N_k - \delta T_k$ plays must occur during dominating rounds. Each dominating round contributes at least $c > 0$ by the minimum reward gap condition. Thus, the regret at $\tau_k$ satisfies:
\[
R^{\phi}_{\tau_k} \ge c \big(N_k - \delta T_k\big) - \delta T_k = c N_k - \delta(c+1)T_k
\]

Notice that the expectation of $N_k$ is exactly the sum of the joint probabilities bounded by our deduction:
\[
\mathbb{E}[N_k] = \sum_{t=1}^{T_k} \Pr(a^t = a \cap D^t) > \epsilon' T_k
\]
Taking expectations on $R^{\phi}_{\tau_k}$, and substituting $\mathbb{E}[N_k] > \epsilon' T_k$, yields:
\[
\mathbb{E}[R^{\phi}_{\tau_k}] \ge \left( c\epsilon' - \delta(c+1) \right) T_k
\]

By our initial choice of $\delta = \frac{c\epsilon'}{2(c+1)}$, we can define the constant $C = c\epsilon' - \delta(c+1) > 0$, yielding $\mathbb{E}[R^{\phi}_{\tau_k}] \ge C T_k > 0$.

Because the instantaneous reward difference is bounded by $1$, the specific regret trajectory $R^\phi_t$ is $1$-Lipschitz. Consequently, for any deterministic time $t$ and any realized trajectory of actions and rewards, the regret satisfies $R^\phi_t \ge R^\phi_{\tau_k} - |t - \tau_k|$.

By definition, the overall swap regret dominates any specific swap function, and trivially dominates the identity swap (yielding zero regret). Therefore, $R_{\text{swap}}(t) \ge \max\big(0, R^\phi_t\big)$. Substituting the bound yields:
\[
R_{\text{swap}}(t) \ge \max\big(0, R^\phi_{\tau_k} - |t - \tau_k|\big)
\]

The maximum expected swap regret over the horizon $[1 \dots T_k]$ is bounded below by its time average:
\[
\max_{t \in [1 \dots T_k]} \mathbb{E}[R_{\text{swap}}(t)] \ge \frac{1}{T_k} \sum_{t=1}^{T_k} \mathbb{E}[R_{\text{swap}}(t)] \ge \frac{1}{T_k} \sum_{t=1}^{T_k} \mathbb{E}\left[ \max\big(0, R^\phi_{\tau_k} - |t - \tau_k|\big) \right]
\]

To lower-bound the sum on the right-hand side, we define the non-negative peak regret $R^+_{\tau_k} = \max(0, R^\phi_{\tau_k})$. Consider the local history interval $[\tau_k - \lfloor R^+_{\tau_k} \rfloor, \tau_k]$. Because the total regret cannot exceed the round count, $R^+_{\tau_k} \le \tau_k$, guaranteeing this interval is fully contained within $[1 \dots T_k]$.

Within this interval, the sequence $\max\big(0, R^\phi_{\tau_k} - (\tau_k - t)\big)$ forms a decreasing arithmetic progression. The sum of this sequence is bounded below by $\frac{1}{2}(R^+_{\tau_k})^2$. Applying this bound and Jensen's inequality ($\mathbb{E}[X^2] \ge (\mathbb{E}[X])^2$) yields:
\[
\max_{t \in [1 \dots T_k]} \mathbb{E}[R_{\text{swap}}(t)] \ge \frac{1}{T_k}\mathbb{E}\left[ \frac{1}{2}(R^+_{\tau_k})^2 \right] \ge \frac{(\mathbb{E}[R^+_{\tau_k}])^2}{2T_k}
\]
Since $R^+_{\tau_k} \ge R^\phi_{\tau_k}$, its expectation is bounded below by our earlier deduction: $\mathbb{E}[R^+_{\tau_k}] \ge \mathbb{E}[R^\phi_{\tau_k}] \ge C T_k$. Substituting this yields:
\[
\max_{t \in [1 \dots T_k]} \mathbb{E}[R_{\text{swap}}(t)] \ge \frac{(C T_k)^2}{2T_k} = \frac{C^2}{2} T_k
\]
Let $t^*_k = \arg\max_{t \in [1 \dots T_k]} \mathbb{E}[R_{\text{swap}}(t)]$ denote the deterministic time step achieving this maximum. This directly establishes $\mathbb{E}[R_{\text{swap}}(t^*_k)] \ge \frac{C^2}{2} T_k$.

Because $t^*_k \le T_k$, the expected time-averaged swap regret evaluated strictly at this deterministic horizon satisfies:
\[
\frac{1}{t^*_k} \mathbb{E}[R_{\text{swap}}(t^*_k)] \ge \frac{1}{t^*_k} \left( \frac{C^2}{2} T_k \right) \ge \frac{C^2}{2} > 0
\]
Furthermore, because maximal expected regret is bounded by time ($\mathbb{E}[R_{\text{swap}}(t^*_k)] \le t^*_k$), this inequality implies $t^*_k \ge \frac{C^2}{2} T_k$. As $T_k \to \infty$, the extracted sequence of deterministic horizons $t^*_k \to \infty$. Evaluating the limit superior over this sequence yields:
\[
\limsup_{T \to \infty} \frac{1}{T} \mathbb{E}\left[ R_{\text{swap}}(T) \right] \ge \limsup_{k \to \infty} \frac{1}{t^*_k} \mathbb{E}\left[ R_{\text{swap}}(t^*_k) \right] \ge \frac{C^2}{2} > 0
\]
Because No-Swap-Regret requires this limit superior to be bounded by zero against all adversaries, this contradiction completes the proof.
\end{proof}

\subsection{Contextual Learning}

Many practical learning algorithms condition their behavior on an observable state or context. For example, the Q-learning agents in \cite{CCDP20} condition their pricing decisions on the prices played in the previous round. A natural question is whether introducing contextual memory breaks the Domination-Avoiding guarantee. We establish that the DA class is closed under finite contexts: an agent that behaves as a DA learner within each specific context remains globally Domination-Avoiding.

\begin{definition}[Contextual Domination-Avoiding]
    Let $\mathcal{C}$ be a finite set of contexts. A Contextual Domination-Avoiding (C-DA) agent maintains an independent instance of a Domination-Avoiding (DA) learning algorithm for each context $c \in \mathcal{C}$. At each time step $t$, an adversary arbitrarily (and potentially adaptively) selects a context $c^t \in \mathcal{C}$. The C-DA agent queries the DA instance associated with $c^t$ to select its action $a^t$, and exclusively updates this active instance using the realized reward $r^t(a^t)$.
\end{definition}

\begin{customthm}{4.5} \label{thm:contextual_da}
    Any Contextual Domination-Avoiding (C-DA) agent operating over a finite context space $\mathcal{C}$ is a global Domination-Avoiding (DA) agent.
\end{customthm}

\begin{proof}
Fix an arbitrary $\epsilon > 0$. We will show there exists $\delta > 0$ such that $\operatorname{UTA}(\{d^t(\delta)\}) \le \epsilon$.

Let $|\mathcal{C}| = K$. Each context $c \in \mathcal{C}$ is governed by an independent DA instance. By definition, for any $\epsilon^* > 0$ and each context $c$, there exists a threshold $\delta_c > 0$ such that the Upper-Time-Average (UTA) of the instance playing a $\delta_c$-empirically dominated action (evaluated exclusively on its internal history) is at most $\epsilon^*$. We set $\epsilon^* = \frac{\epsilon}{4K}$ and define the uniform local threshold $\delta^* = \min_{c \in \mathcal{C}} \delta_c > 0$. We then define the global threshold $\delta = \delta^* \frac{\epsilon}{4K}$.

Suppose action $a$ is globally $\delta$-dominated by $a'$ up to time $t$. The total number of global exception rounds where $r^s(a) \ge r^s(a')$ is at most $\delta t$. Let $N^t(c)$ denote the number of times context $c$ occurs up to time $t$. The exception rounds realized while context $c$ is active are a subset of the global exceptions, bounded by $\delta t$.

We partition the contexts at any global time $t$ into two regimes: \emph{frequent} and \emph{rare}. Formally, we define a context $c$ as frequent at time $t$ if $N^t(c) > \frac{\epsilon}{4K} t$, and rare if $N^t(c) \le \frac{\epsilon}{4K} t$. 

For any frequent context $c$, rearranging its bounding inequality yields $t < \frac{4K}{\epsilon} N^t(c)$. Substituting this bounds the local exception rounds in terms of the internal clock:
\[
\left| \left\{ s \le t \mid c^s = c \text{ and } r^s(a) \ge r^s(a') \right\} \right| \le \delta t < \delta \left( \frac{4K}{\epsilon} N^t(c) \right) = \delta^* N^t(c)
\]
Thus, if $a$ is globally $\delta$-dominated at global time $t$, it is $\delta^*$-dominated within the local history of any frequent context $c$.

Recall that $d^t(\delta)$ denotes the marginal probability that the C-DA agent plays a globally $\delta$-dominated action at time $t$. We expand this probability by partitioning over the active context $c^t$ and its realization frequency $N^t(c)$. Dropping the domination condition for rare contexts and substituting local domination for frequent contexts yields:
\begin{align*}
d^t(\delta) &= \sum_{c \in \mathcal{C}} \Pr(a^t \text{ is globally } \delta\text{-dominated} \land c^t = c) \\
&\le \sum_{c \in \mathcal{C}} \Pr\left(N^t(c) \le \tfrac{\epsilon}{4K} t \land c^t = c \right) + \sum_{c \in \mathcal{C}} \Pr(a^t \text{ is locally } \delta^*\text{-dominated} \land c^t = c)
\end{align*}

We evaluate the time average over the global horizon $T$.
For any context $c \in \mathcal{C}$ and any realized sample path, the condition $(N^t(c) \le \frac{\epsilon}{4K} t \land c^t = c)$ can occur at most $\lfloor \frac{\epsilon}{4K} T \rfloor$ times. Therefore, the sum of the marginal probabilities over the horizon is deterministically bounded:
\begin{align*}
\limsup_{T \to \infty} \frac{1}{T} \sum_{t=1}^T \sum_{c \in \mathcal{C}} \Pr \Big( N^t(c) \le \tfrac{\epsilon}{4K} t \land c^t = c \Big) &\le \limsup_{T \to \infty}\sum_{c \in \mathcal{C}} \frac{1}{T} \lfloor \tfrac{\epsilon}{4K} T \rfloor \\
&\le K \left( \frac{\epsilon}{4K} \right) = \frac{\epsilon}{4}
\end{align*}

For any frequent context $c$, the time-averaged probability of playing a locally $\delta^*$-dominated action over the global horizon $T$ corresponds to the expected number of errors the local instance makes up to its random internal horizon $N^T(c)$, scaled by $1/T$. Because the local clock satisfies $N^T(c) \le T$ and error probabilities are non-negative, we can upper-bound this random sum by evaluating the local instance over a full deterministic horizon $T$ against an extended adversary. 

Let $d_c^\tau(\delta^*)$ denote the expected marginal probability that instance $c$ plays a $\delta^*$-dominated action at its internal time step $\tau$ against this extended sequence. Bounding the random horizon yields:
\[
\frac{1}{T} \sum_{t=1}^T \Pr(a^t \text{ is locally } \delta^*\text{-dominated} \land c^t = c) = \frac{1}{T} \sum_{\tau=1}^{N^T(c)} d_c^\tau(\delta^*)  \le \frac{1}{T} \sum_{\tau=1}^T d_c^\tau(\delta^*)
\]
Taking the limit superior, the Domination-Avoiding property of the local instance $c$ guarantees the sequence is bounded:
\[
\limsup_{T \to \infty} \frac{1}{T} \sum_{t=1}^T \Pr(a^t \text{ is locally } \delta^*\text{-dominated} \land c^t = c) \le \limsup_{T \to \infty} \frac{1}{T} \sum_{\tau=1}^T d_c^\tau(\delta^*) \le \frac{\epsilon}{4K}
\]
Summing this bound over all $K$ contexts yields:
\[
\limsup_{T \to \infty} \sum_{c \in \mathcal{C}} \frac{1}{T} \sum_{t=1}^T \Pr(a^t \text{ is locally } \delta^*\text{-dominated} \land c^t = c) \le K \left(\frac{\epsilon}{4K}\right) = \frac{\epsilon}{4}
\]

Combining the asymptotic bounds of both regimes yields $\operatorname{UTA}(\{d^t(\delta)\}) \le \frac{\epsilon}{4} + \frac{\epsilon}{4} = \frac{\epsilon}{2} < \epsilon$.
\end{proof}

\section{Deferred Proofs from Section 2}
\label{app:sec2}

\begin{lemma} \label{lem:da_equivalence}
    The uniform sampling formulation of a Domination-Avoiding agent (\cref{def:da}) is equivalent to the Upper-Time-Average formulation (\cref{def:alt_da}).
\end{lemma}

\begin{proof}
    \cref{def:da} requires that for any adversary and any $\epsilon > 0$, there exists $\delta > 0$ such that for all sufficiently large $T$:
    \[
    \Pr_{t \sim {U}(1 \dots T)}\left[a^t \text{ is } \delta\text{-dominated} \text{ up to time } t \right] \le \epsilon
    \]
    Let $d^t(\delta)$ denote the marginal probability that the agent plays a $\delta$-dominated action at time $t$. Because $t$ is drawn uniformly from $\{1, \dots, T\}$, this probability evaluates to the finite-time average:
    \[
    \frac{1}{T} \sum_{t=1}^T d^t(\delta) \le \epsilon
    \]
    Assume the agent satisfies \cref{def:da}. Fix an arbitrary $\epsilon > 0$. There exists $\delta > 0$ such that the finite-time average is bounded by $\epsilon$ for all sufficiently large $T$. Taking the limit superior on both sides preserves this non-strict inequality directly, yielding:
    \[
    \limsup_{T \to \infty} \frac{1}{T} \sum_{t=1}^T d^t(\delta) \le \epsilon
    \]
    Substituting the Upper-Time-Average notation gives $\operatorname{UTA}(\{d^t(\delta)\}) \le \epsilon$, satisfying \cref{def:alt_da}.

    Conversely, assume the agent satisfies \cref{def:alt_da}. Fix an arbitrary $\epsilon > 0$. Invoking the definition at a tighter threshold of $\epsilon / 2 > 0$, there exists $\delta > 0$ such that:
    \[
    \limsup_{T \to \infty} \frac{1}{T} \sum_{t=1}^T d^t(\delta) \le \frac{\epsilon}{2}
    \]
    By the definition of the limit superior, if the limit superior of a sequence is bounded by $\epsilon/2$, the sequence itself can exceed $\epsilon$ at most finitely many times. Therefore, for all sufficiently large $T$, the finite-time average must be at most $\epsilon$, satisfying \cref{def:da}.
\end{proof}

\section{Deferred Proofs from Section 3}
\label{app:sec3}

The core intuition of Theorem 1 is to construct a symmetric bimodal distribution over a moderate base price and an extreme bonus price. By precisely choosing the mixing probabilities, we guarantee that any unilateral deviation yields a strictly lower expected payoff than the equilibrium strategy.

\begin{customthm}{1$'$}
    Consider the normalized Bertrand Logit game ($\alpha = 0$, $c = 0$, $\mu = 1$).
    For any target
    price $V > p^*$, where $p^*$ is the competitive duopoly equilibrium price (of the continuous Bertrand game), there exists a Coarse Correlated Equilibrium (CCE) that is supported on
    strategies where both agents always bid 
    a price $p_i \ge V$.
\end{customthm}

Proving the stability of this equilibrium requires bounding the maximum utility a player can extract by unilaterally deviating. The following lemma establishes this structural ceiling, demonstrating that the maximum achievable profit against any supra-competitive opponent price remains bounded strictly below that price.

\begin{lemma}[Profit Gap]\label{lem:profit_gap}
    In the normalized market ($\alpha = 0$, $c = 0$, $\mu = 1$), for any opponent price $p_{-i} > p^* = 2$, the maximum achievable profit $u_{max}(p_{-i}) = \max_{p_i \ge 0} u_i(p_i, p_{-i})$ satisfies:
    \[
    u_{max}(p_{-i}) < p_{-i} - 1
    \]
\end{lemma}

\begin{proof}
Substituting our parameters, the Logit profit function simplifies to $u_i(p_i, p_{-i}) = \frac{p_i}{1+e^{p_i-p_{-i}}}$.
The partial derivative with respect to $p_i$ is:
\[
\frac{\partial u_i}{\partial p_i} = \frac{1 + e^{p_i - p_{-i}}(1 - p_i)}{(1+e^{p_i-p_{-i}})^2}
\]
Evaluating this marginal profit where the agent matches the opponent's price ($p_i = p_{-i}$), the denominator is strictly positive and the numerator simplifies to $2 - p_{-i}$.
Since the opponent's price is supra-competitive ($p_{-i} > 2$), this marginal profit is strictly negative. Consequently, the profit function is decreasing at the matching price, implying the optimal response requires undercutting ($p_i < p_{-i}$).

Fixing the opponent's price $p_{-i}$, we define the strictly positive undercut amount as $x = p_{-i} - p_i > 0$.  We can now express the agent's profit as a function of the undercut:
\[
u_i(x) = \frac{p_{-i} - x}{1+e^{-x}}
\]
Because the Logit profit function is strictly unimodal, the unique root of the first-order condition yields the global maximum:
\[
\frac{d u_i}{dx} = \frac{-(1+e^{-x}) - (p_{-i} - x)(-e^{-x})}{(1+e^{-x})^2} = 0
\]
Setting the numerator to zero and rearranging provides an identity for the optimal undercut $x^*$:
\begin{equation} \label{eq:undercut}
p_{-i} - x^* = \frac{1+e^{-x^*}}{e^{-x^*}} = e^{x^*} + 1
\end{equation}
We calculate the maximum achievable profit by substituting \cref{eq:undercut} back into the profit equation:
\[
u_{max}(p_{-i}) = \frac{p_{-i} - x^*}{1+e^{-x^*}} = \frac{e^{x^*}+1}{1+e^{-x^*}} = e^{x^*}
\]
Using the identity $p_{-i} = x^* + e^{x^*} + 1$ derived from \cref{eq:undercut}, we substitute $u_{max}(p_{-i})$:
\[
p_{-i} = x^* + u_{max}(p_{-i}) + 1 \implies p_{-i} - u_{max}(p_{-i}) = x^* + 1
\]
Since the optimal strategy requires undercutting ($x^* > 0$), it follows that $p_{-i} - u_{max}(p_{-i}) > 1$. Because this continuous bound holds globally, it applies to any restricted discrete action space $\mathcal{P}$.
\end{proof}

With this bound on deviation profit established, we proceed to the main construction of the CCE distribution.

\begin{proof}[Proof of Theorem 1$'$]
Let $V > p^*$ be the target price. We construct a symmetric CCE supported on two prices: a \textit{base price} $p_L$ and a \textit{bonus price} $p_H$. We define the base price as $p_L = V$. Because $V > p^* = 2$, this ensures $p_L$ is supra-competitive. We then fix a bonus price $p_H \gg p_L$ such that:
\begin{equation} \label{eq:bound_exponential}
p_H^2 e^{-\frac{p_H}{p_L}} \le p_L^2 e^{-p_L}
\end{equation}
The existence of such a $p_H$ is guaranteed since the exponential decay term $e^{-p_H/p_L}$ asymptotically dominates the polynomial growth of $p_H^2$ as $p_H \to \infty$.

We define the symmetric joint strategy distribution $\sigma_{CCE}$ as:
\[
\sigma_{CCE} = \begin{cases}
(p_L, p_L) & \text{w.p.} \quad q \\
(p_H, p_H) & \text{w.p.} \quad 1-q
\end{cases}
\]
To calibrate the equilibrium, we select the mixing probability $q$ such that the expected payoff contributions from the two states are balanced: $q p_L = (1-q) p_H$. Solving for $q$ yields $q = \frac{p_H}{p_L + p_H}$. The expected equilibrium payoff is therefore $U_{CCE} = q p_L = (1-q) p_H$.

By construction, $p_H > p_L = V$. Therefore, the distribution $\sigma_{CCE}$ is supported exclusively on prices $p_i \ge V$.

To complete the proof, we must verify that $\sigma_{CCE}$ is a valid CCE. This requires that the expected utility of any unilateral deviation to \textbf{any} price $p$ does not exceed the equilibrium payoff $U_{CCE}$:
\[
U_{dev}(p) = q u_i(p, p_L) + (1-q) u_i(p, p_H) \le U_{CCE}
\]
By substituting our indifference condition $U_{CCE} = q p_L = (1-q) p_H$ into the right side and dividing by either $(1-q)$ or $q$, we rearrange this stability requirement into two mathematically equivalent forms:
\begin{align}
    \textbf{Form 1:} \quad \frac{p_L}{p_H} u_i(p, p_H) &\le p_L - u_i(p, p_L) \tag{I} \\
    \textbf{Form 2:} \quad \frac{p_H}{p_L} u_i(p, p_L) &\le p_H - u_i(p, p_H) \tag{II}
\end{align}
We verify stability by splitting the deviation price space into two continuous regions at the threshold $\hat{p} = \frac{p_H}{p_L}$.

\paragraph{Case 1: Low Deviations ($p \le \frac{p_H}{p_L}$)}
We evaluate stability using \textbf{Form 1}. By \cref{lem:profit_gap}, the profit gap against a static supra-competitive price $p_L$ is globally bounded, yielding $p_L - u_i(p, p_L) > 1$. Thus, the right-hand side is strictly greater than $1$.

For the left-hand side, we use the trivial utility bound $u_i(p, p_H) \le p$. Because we are evaluating the region where $p \le \hat{p}$, we can strictly bound the left-hand side:
\[
\frac{p_L}{p_H} u_i(p, p_H) \le \frac{p_L}{p_H} p \le \frac{p_L}{p_H} \left( \frac{p_H}{p_L} \right) = 1
\]
Chaining these inequalities yields $\frac{p_L}{p_H} u_i(p, p_H) \le 1 < p_L - u_i(p, p_L)$, confirming that any deviation in this lower range is strictly unprofitable.

\paragraph{Case 2: High Deviations ($p > \frac{p_H}{p_L}$)}
We evaluate stability using \textbf{Form 2}. By applying \cref{lem:profit_gap} to the bonus price $p_H$, the right-hand side is globally bounded below: $p_H - u_i(p, p_H) > 1$.

For the left-hand side, we bound the utility function: $u_i(p, p_L) < p e^{-(p-p_L)}$. The function $f(p) = p e^{-(p-p_L)}$ has derivative $f'(p) = e^{-(p-p_L)}(1 - p)$, which is negative for all $p > 1$. Since we constructed $p_H \gg p_L$, our threshold $\hat{p} = \frac{p_H}{p_L}$ is strictly greater than $1$. Thus, $f(p)$ is strictly decreasing across this upper region, and its maximum occurs exactly at the lower boundary $p = \hat{p}$. Substituting $\hat{p}$ into the left-hand side:
\[
\frac{p_H}{p_L} u_i(p, p_L) < \frac{p_H}{p_L} \left( \hat{p} e^{-(\hat{p} - p_L)} \right) = \frac{p_H}{p_L} \left( \frac{p_H}{p_L} \right) e^{-\left(\frac{p_H}{p_L} - p_L\right)} = \frac{p_H^2}{p_L^2} e^{-\frac{p_H}{p_L} + p_L}
\]
To guarantee this left-hand bound is at most $1$, we multiply both sides by $p_L^2 e^{-p_L}$, yielding the requirement:
\[
p_H^2 e^{-\frac{p_H}{p_L}} \le p_L^2 e^{-p_L}
\]
This is exactly the bound from \cref{eq:bound_exponential} we enforced during our construction step. Therefore:
\[
\frac{p_H}{p_L} u_i(p, p_L) \le 1 < p_H - u_i(p, p_H)
\]
Since unilateral deviations in both continuous price regions are unprofitable, $\sigma_{CCE}$ constitutes a valid Coarse Correlated Equilibrium entirely supported on prices $p_i \ge V$.

Because $\sigma_{CCE}$ is robust against any continuous deviation $p \ge 0$, it naturally remains a valid CCE over any finite grid $\mathcal{P}$ that contains $p_L$ and $p_H$.
\end{proof}

\section{Deferred Proofs from Section 4}
\label{app:sec4}

\subsection{General Convergence to IESPDS}

\begin{customlemma}{4.1}
    Suppose all $\mathcal{N}$ players use DA agents to repeatedly play the finite normal-form stage game $\Gamma$. For any player $i$, any surviving action $a_i \in S_i$, and every pure action $a'_i \in A_i$, there exists a surviving opponent profile $a_{-i} \in S_{-i}$ such that $u_i(a_i, a_{-i}) \ge u_i(a'_i, a_{-i})$.
\end{customlemma} 

\begin{proof}
We evaluate the probability of the opponents playing a non-surviving joint profile, denoted as $a_{-i}^t \notin S_{-i}$. This event occurs if and only if at least one opponent $j \neq i$ plays a non-surviving action $a_j \notin S_j$. By the union bound, the probability that the opponents play a non-surviving profile at time $t$ is bounded by the sum of the marginal probabilities of the opponents playing non-surviving actions:
\[
\Pr(a_{-i}^t \notin S_{-i}) \le \sum_{j \neq i} \sum_{a_j \notin S_j} p_j^t(a_j)
\]
By the definition of a non-surviving action, for every $a_j \notin S_j$, we have $\operatorname{UTA}(\{p_j^t(a_j)\}) = 0$. Because the limit superior is subadditive and our game has a finite number of players and actions, the limit superior of this finite sum is bounded by the sum of their individual zero limits:
\[
\limsup_{T \to \infty} \frac{1}{T} \sum_{t=1}^T \Pr(a_{-i}^t \notin S_{-i}) = 0
\]
Since probabilities are non-negative, the limit exists and evaluates to $0$.

Assume, for contradiction, that there exists a surviving action $a_i \in S_i$ and a pure action $a'_i \in A_i$ such that $u_i(a_i, a_{-i}) < u_i(a'_i, a_{-i})$ against \textbf{all} surviving opponent profiles $a_{-i} \in S_{-i}$. Let $c = \operatorname{UTA}(\{p_i^t(a_i)\}) > 0$ denote its survival frequency.
Because $a'_i$ strictly dominates $a_i$ against all surviving profiles, the event $u_i(a_i, a_{-i}^t) \ge u_i(a'_i, a_{-i}^t)$ can only occur if the opponents play a non-surviving profile. Let $I_i^t = \mathbbm{1}\{u_i(a_i, a_{-i}^t) \ge u_i(a'_i, a_{-i}^t)\}$ be the instantaneous indicator. It follows that $I_i^t \le \mathbbm{1}\{a_{-i}^t \notin S_{-i}\}$.

Let $D^t$ be the event that $a_i$ is $\delta$-empirically dominated by $a'_i$ up to time $t$. The complement event at some time $T$, denoted $\neg D^T$, requires that $a_i$ performed at least as well as $a'_i$ strictly more than $\delta T$ times (i.e., $\sum_{t=1}^T I_i^t > \delta T$). We bound the probability of this complement using Markov's inequality:
\[
\Pr(\neg D^T) = \Pr\left( \sum_{t=1}^T I_i^t > \delta T \right) \le \frac{\mathbb{E}\left[\sum_{t=1}^T I_i^t\right]}{\delta T} \le \frac{1}{\delta T} \sum_{t=1}^T \Pr(a_{-i}^t \notin S_{-i})
\]
From our earlier derivation, the time-average on the right-hand side converges to $0$. Thus, for any fixed $\delta > 0$, the sequence satisfies $\lim_{T \to \infty} \Pr(\neg D^T) = 0$.

Recall that $d_i^t(\delta)$ represents the overall probability of player $i$ playing an action that is currently $\delta$-dominated at time $t$. By bounding the joint probability using the complement event $\neg D^t$, we obtain:
\[
d_i^t(\delta) \ge \Pr(a_i^t = a_i \land D^t) \ge p_i^t(a_i) - \Pr(\neg D^t)
\]
Averaging over time $T$ and taking the limit superior on both sides yields:
\[
\operatorname{UTA}(\{d_i^t(\delta)\}) \ge \operatorname{UTA}(\{p_i^t(a_i)\}) - \limsup_{T \to \infty} \frac{1}{T} \sum_{t=1}^T \Pr(\neg D^t)
\]
Since the sequence $\Pr(\neg D^t)$ converges to $0$, its time-average also converges to $0$. Substituting this limit and our survival frequency $c$ yields:
\[
\operatorname{UTA}(\{d_i^t(\delta)\}) \ge c - 0 = c
\]
This lower bound holds for all $\delta > 0$. However, because player $i$ is a DA agent, choosing an error threshold of $\epsilon = c / 2 > 0$ implies there exists some $\delta > 0$ such that $\operatorname{UTA}(\{d_i^t(\delta)\}) \le c / 2$. This yields the contradiction $c \le c/2$. Therefore, there must exist some $a_{-i} \in S_{-i}$ such that $u_i(a_i, a_{-i}) \ge u_i(a'_i, a_{-i})$.
\end{proof}

With Lemma \ref{lem:surviving_undominated_N} established, we have a strict local condition: an action can only survive if it remains undominated with respect to the surviving profiles of the opponents. We now apply this structural constraint recursively to prove \cref{thm:da_IESPDS}.

\begin{customthm}{2}
    Consider any finite $\mathcal{N}$-player game and any $\mathcal{N}$ Domination-Avoiding agents that repeatedly play this game against each other.
    For every action $a_i$ that is removable by Iterated Elimination of Strictly Purely Dominated Strategies we have that 
    \[
    \lim_{T\rightarrow \infty} \Pr_{t \sim U(1 \dots T)} [\text{player } i \text{ plays } a_i \text{ at time } t]=0
    \]
    where the probability is taken over a  {uniform} random choice of a time $t$ in the first $T$ time steps as well as the internal randomization of the agents.
\end{customthm}

\begin{proof}
Recall that $S_i^k$ denotes the set of actions for player $i$ that survive after $k$ rounds of Iterated Elimination of Strictly Purely Dominated Strategies (IESPDS). We proceed by induction on the elimination depth $k$ to show that the set of empirically surviving actions $S_i$ is a subset of $S_i^k$ for all $k \ge 0$ and all players $i \in \mathcal{N}$.

\textbf{Base Case ($k=0$):} By definition, the initial IESPDS set is the entire action space, $S_i^0 = A_i$. Since empirically surviving actions are inherently a subset of all available actions, $S_i \subseteq S_i^0$ holds trivially for all $i \in \mathcal{N}$.

\textbf{Inductive Step:} Assume that for some $k \ge 0$, $S_i \subseteq S_i^k$ holds for all players $i \in \mathcal{N}$. This implies that the joint empirically surviving profiles of the opponents satisfy $S_{-i} \subseteq S_{-i}^k$. We must show $S_i \subseteq S_i^{k+1}$ for an arbitrary player $i$.
Consider any empirically surviving action $a_i \in S_i$. By the inductive hypothesis, $a_i \in S_i^k$. By \cref{lem:surviving_undominated_N}, for every pure action $a'_i \in A_i$, there exists a surviving opponent profile $a_{-i} \in S_{-i}$ such that $u_i(a_i, a_{-i}) \ge u_i(a'_i, a_{-i})$. Since $S_{-i} \subseteq S_{-i}^k$, this specific profile $a_{-i}$ is also in $S_{-i}^k$. Therefore, there does not exist any single action $a'_i$ that strictly dominates $a_i$ against all opponent profiles in $S_{-i}^k$. By the recursive definition of IESPDS, $a_i$ cannot be eliminated at step $k$, implying $a_i \in S_i^{k+1}$.

By induction, $S_i \subseteq S_i^k$ for all $k \ge 0$. Because the stage game $\Gamma$ has finite action spaces, the elimination process must stabilize after a finite number of steps $M$, yielding the final stable set $S_i^\infty = S_i^M$. Since $S_i \subseteq S_i^\infty$, any action $a_i$ removed by IESPDS satisfies $a_i \notin S_i^\infty$. Consequently, it cannot be an empirically surviving action ($a_i \notin S_i$), meaning its Upper-Time-Average evaluates to zero: 
\[
\operatorname{UTA}(\{p_i^t(a_i)\}) = \limsup_{T \to \infty} \frac{1}{T} \sum_{t=1}^T p_i^t(a_i) = 0
\]
Since probabilities are non-negative, a limit superior of zero implies that the sequence of finite-time averages converges to exactly zero. Because the probability of playing $a_i$ at a uniformly chosen time step $t \in \{1 \dots T\}$ evaluates precisely to this finite-time average (by the law of total probability), we obtain:
\[
\lim_{T\rightarrow \infty} \Pr_{t \sim {U}(1 \dots T)} [\text{player } i \text{ plays } a_i \text{ at time } t] = 0
\]
\end{proof}

\subsection{IESPDS Characterization for Bertrand Logit}

\begin{customthm}{5}
    Consider a symmetric Bertrand Logit pricing game over a finite $\epsilon$-dense grid $\mathcal{P}$ of possible price levels. Let $p^*$ be the Nash equilibrium of the continuous game, and assume $\min(\mathcal{P}) \le p^* \le \max(\mathcal{P})$.
    Then only prices $p \in \mathcal{P}$ satisfying $|p-p^*| \le 2 \epsilon$ can survive Iterated Elimination of Strictly Purely Dominated Strategies (IESPDS).
\end{customthm}

To execute the systematic elimination described, we must define the boundary conditions under which non-competitive prices become strictly dominated. We define two continuous threshold functions, parameterized by the maximum grid density $\epsilon$:
\begin{align}
    \Phi^+(p) &= c + \frac{\epsilon}{1-e^{-\mu\epsilon}} \left[ 1 + \frac{1}{1 + \alpha e^{\mu p}} \right] \label{eq:phi_upper} \\
    \Phi^-(p) &= c + \frac{\epsilon}{e^{\mu\epsilon} - 1} \left[ 1 + \frac{1}{1 + \alpha e^{\mu p}} \right] \label{eq:phi_lower}
\end{align}
Because both $\Phi^+(p)$ and $\Phi^-(p)$ are continuous non-increasing and bounded for all $p \ge 0$, each possesses a unique fixed point, which we denote as $p^+$ and $p^-$ respectively (i.e., $p^+=\Phi^+(p^+)$ and $p^-=\Phi^-(p^-)$).

\begin{lemma}[Bounds on the Nash equilibrium]\label{lem:bounds}
The values $p^-$ and $p^+$ bound the equilibrium $p^*$ from
below and from above and are close to each other.  Specifically, $p^- < p^* < p^+ \le p^-+2\epsilon$.
\end{lemma}

\begin{proof}
For any strictly positive grid gap $\epsilon > 0$, standard exponential inequalities ($e^x > 1+x$) dictate that the leading coefficients of our threshold functions strictly bound the continuous inverse price sensitivity $\frac{1}{\mu}$:
\[
\frac{\epsilon}{e^{\mu\epsilon} - 1} < \frac{1}{\mu} < \frac{\epsilon}{1 - e^{-\mu\epsilon}}
\]
Because the bracketed demand terms in \cref{eq:phi_upper,eq:phi_lower} are identical and strictly positive, this coefficient ordering ensures that the fixed points bound the continuous Nash equilibrium: $p^- < p^* < p^+$.

To evaluate the maximum geometric distance, we bound the absolute width of this continuous interval:
\[
p^+ - p^- = \Phi^+(p^+) - \Phi^-(p^-)
\]
Because $\Phi^+(p)$ is non-increasing and $p^+ > p^-$, we can construct a valid upper bound by evaluating the first term at $p^-$ instead:
\[
p^+ - p^- \le \Phi^+(p^-) - \Phi^-(p^-)
\]
Substituting \cref{eq:phi_upper,eq:phi_lower}, the marginal cost constants ($c$) cancel. Factoring out the common bracketed demand term yields:
\[
p^+ - p^- \le \epsilon \left[ \frac{1}{1-e^{-\mu\epsilon}} - \frac{1}{e^{\mu\epsilon}-1} \right] \left( 1 + \frac{1}{1 + \alpha e^{\mu p^-}} \right)
\]
Multiplying the numerator and denominator of the first internal fraction by $e^{\mu\epsilon}$ yields $\frac{e^{\mu\epsilon}}{e^{\mu\epsilon}-1}$. Subtracting the second fraction, the term inside the square brackets simplifies to $1$. Furthermore, because our model permits any relative outside option strength $\alpha \ge 0$, the remaining bracketed demand term evaluates to at most $2$. Therefore:
\[
p^+ - p^- \le 2\epsilon \implies p^+ \le p^- + 2\epsilon
\]
\end{proof}

Our iterative elimination process relies on two symmetric lemmas. The first establishes that if the opponent's maximum possible price is bounded, the highest available grid prices become strictly dominated by undercutting.

\begin{lemma}[Strict Dominance of Undercutting] \label{lem:undercutting}
    Let $p^{(j)} \in \mathcal{P}$ be a grid price such that $p^{(j)} > p^+$. The pure strategy of playing $p^{(j)}$ is strictly dominated by undercutting to the adjacent lower price $p^{(j-1)}$, against any opponent pure strategy bounded by $p_{-i} \le p^{(j)}$.
\end{lemma}

Symmetrically, the second lemma establishes that if the opponent's minimum possible price is bounded, the lowest available grid prices become strictly dominated by up-pricing.

\begin{lemma}[Strict Dominance of Up-Pricing] \label{lem:up_pricing}
    Let $p^{(j)} \in \mathcal{P}$ be a grid price such that $p^{(j)} < p^-$. The pure strategy of playing $p^{(j)}$ is strictly dominated by up-pricing to the adjacent higher price $p^{(j+1)}$, against any opponent pure strategy bounded by $p_{-i} \ge p^{(j)}$.
\end{lemma}

We defer the proofs of these two boundary lemmas to the end of this section. Armed with these elimination conditions, we now prove our main convergence result.

\begin{proof}[Proof of \cref{thm:IESPDS_collapse}]
We proceed by a two-sided induction on the finite, ordered discrete price grid $\mathcal{P} = \{p^{(1)}, p^{(2)}, \dots, p^{(M)}\}$. Because the IESPDS survival set in finite games is order-independent, we can analyze the elimination in two distinct directional phases.

\paragraph{Top-Down Elimination:}
If the maximum grid price already satisfies $p^{(M)} \le p^+$, then the upper bound of our survival interval is trivially satisfied and no top-down elimination is required. 
Otherwise, we have $p^{(M)} > p^+$. In the initial game, the opponent's maximum possible price is naturally restricted by the grid boundary: $p_{-i} \le p^{(M)}$. By \cref{lem:undercutting}, the pure strategy $p^{(M)}$ is strictly dominated by $p^{(M-1)}$ and is eliminated. 

Proceeding inductively, let the highest surviving grid price at any step be $p^{(j)}$. As long as $p^{(j)} > p^+$, the opponent's surviving strategy space is restricted to $p_{-i} \le p^{(j)}$. \cref{lem:undercutting} guarantees $p^{(j)}$ is strictly dominated by $p^{(j-1)}$. This elimination cascades downward, halting only when the highest surviving price satisfies $p \le p^+$.

\paragraph{Bottom-Up Elimination:}
Symmetrically, if the minimum grid price already satisfies $p^{(1)} \ge p^-$, then the lower bound is trivially satisfied. 
Otherwise, we have $p^{(1)} < p^-$. The opponent's minimum surviving price is initially restricted to $p_{-i} \ge p^{(1)}$. By \cref{lem:up_pricing}, $p^{(1)}$ is strictly dominated by $p^{(2)}$ and is eliminated. 

Proceeding inductively, let the lowest surviving price be $p^{(j)}$. As long as $p^{(j)} < p^-$, the opponent is restricted to $p_{-i} \ge p^{(j)}$. \cref{lem:up_pricing} guarantees $p^{(j)}$ is strictly dominated by $p^{(j+1)}$. This cascading elimination proceeds upward, halting only when the lowest surviving price satisfies $p \ge p^-$.

The IESPDS process halts, leaving a surviving set of prices restricted strictly to the interval $[p^-, p^+]$. By \cref{lem:bounds}, we have $p^- < p^* < p^+$ and the absolute width of this interval is bounded by $p^+ - p^- \le 2\epsilon$. Because the Nash equilibrium $p^*$ is strictly bounded within this interval, the absolute distance from $p^*$ to any surviving discrete price $p$ cannot exceed the interval's maximum width. Consequently, any surviving price $p \in \mathcal{P}$ must satisfy $|p - p^*| \le 2\epsilon$.
\end{proof}

\begin{remark}[Asymptotic Tightness of the Bound]
    The $2\epsilon$ survival bound established in \cref{thm:IESPDS_collapse} relies on a conservative geometric worst-case: because the true continuous Nash equilibrium $p^*$ must fall somewhere within the surviving interval $[p^-, p^+]$ of maximum width $2\epsilon$, its absolute distance to either boundary is bounded by $2\epsilon$. However, as the discrete grid approaches continuity ($\epsilon \to 0$), this bound effectively halves. A first-order Taylor expansion of the leading threshold coefficients reveals that they spread symmetrically around the continuous inverse price sensitivity $\frac{1}{\mu}$. Specifically, $\frac{\epsilon}{1-e^{-\mu\epsilon}} \approx \frac{1}{\mu} + \frac{\epsilon}{2}$ and $\frac{\epsilon}{e^{\mu\epsilon}-1} \approx \frac{1}{\mu} - \frac{\epsilon}{2}$. Because the threshold boundaries expand symmetrically, $p^*$ inherently approaches the exact midpoint of the surviving interval. Consequently, for sufficiently dense grids, the maximum distance from the equilibrium to any surviving price tightens to a single $\epsilon$, asymptotically confining the IESPDS set to $[p^* - \epsilon, p^* + \epsilon]$.
\end{remark}

We now provide the deferred proofs for the boundary dominance lemmas.

\begin{proof}[Proof of \cref{lem:undercutting}]
Let the actual price gap between the adjacent grid points be $\Delta = p^{(j)} - p^{(j-1)}$. By the definition of our $\epsilon$-dense grid, we know $0 < \Delta \le \epsilon$. We must show that $u_i(p^{(j)} - \Delta, p_{-i}) > u_i(p^{(j)}, p_{-i})$ for all $p_{-i} \le p^{(j)}$.
Substituting the Logit profit function $u_i = (p_i - c) \cdot Q_i$, the required strict inequality is:
\[
(p^{(j)} - \Delta - c) \frac{e^{\mu \cdot (a - p^{(j)} + \Delta)}}{e^{\mu \cdot (a - p^{(j)} + \Delta)} + e^{\mu \cdot (a - p_{-i})} + e^{\mu a_0}} > (p^{(j)} - c) \frac{e^{\mu \cdot (a - p^{(j)})}}{e^{\mu \cdot (a - p^{(j)})} + e^{\mu \cdot (a - p_{-i})} + e^{\mu a_0}}
\]
Dividing the numerator and denominator of both fractions by $e^{\mu \cdot (a - p^{(j)})}$, and letting $Z = e^{\mu \cdot (p^{(j)} - p_{-i})} + \alpha e^{\mu p^{(j)}} > 0$, the inequality becomes:
\[
(p^{(j)} - c - \Delta) \frac{e^{\mu \Delta}}{e^{\mu \Delta} + Z} > (p^{(j)} - c) \frac{1}{1 + Z}
\]
Since both denominators are strictly positive, we cross-multiply and isolate $Z$. Dividing through by $e^{\mu \Delta}$ yields:
\[
Z \left[ (p^{(j)} - c) (1 - e^{-\mu \Delta}) - \Delta \right] > \Delta
\]
Substituting $Z$ back into the inequality, we obtain the fundamental dominance condition:
\begin{equation} \label{eq:undercut_Z}
\left( e^{\mu \cdot (p^{(j)}-p_{-i})} + \alpha e^{\mu p^{(j)}} \right) \left[ (p^{(j)} - c) (1 - e^{-\mu\Delta}) - \Delta \right] > \Delta
\end{equation}
We require this strict inequality to hold for all opponent prices $p_{-i} \le p^{(j)}$. Observe that the left-hand side is strictly decreasing with respect to $p_{-i}$. Therefore, the lower bound of the left-hand side occurs at the upper boundary of the opponent's strategy space, where they match the price exactly ($p_{-i} = p^{(j)}$, yielding $e^0 = 1$). If the strict inequality holds for this worst-case opponent strategy, it holds globally for all $p_{-i} < p^{(j)}$. 
Substituting $p_{-i} = p^{(j)}$ into \cref{eq:undercut_Z} and rearranging to isolate $p^{(j)}$ gives the requirement:
\[
p^{(j)} > c + \frac{\Delta}{1 - e^{-\mu\Delta}} \left[ 1 + \frac{1}{1 + \alpha e^{\mu p^{(j)}}} \right]
\]
Notice that the function $g(x) = \frac{x}{1-e^{-\mu x}}$ is strictly increasing for all $x > 0$. Because our grid gap satisfies $\Delta \le \epsilon$, it strictly follows that $g(\Delta) \le g(\epsilon)$. Therefore, if the inequality holds when evaluated at the maximum possible gap $\epsilon$, it holds for the actual realized gap $\Delta$. Replacing $\Delta$ with $\epsilon$ recovers our upper threshold definition:
\[
p^{(j)} > \Phi^+(p^{(j)})
\]
Finally, we apply our initial assumption that $p^{(j)} > p^+$. Because $\Phi^+(p)$ is non-increasing, this assumption ensures $p^{(j)} > p^+ = \Phi^+(p^+) \ge \Phi^+(p^{(j)})$. Thus, $p^{(j)} > \Phi^+(p^{(j)})$ holds strictly. This satisfies the worst-case boundary condition, which establishes that the pure strategy $p^{(j)}$ is strictly dominated by $p^{(j-1)}$.
\end{proof}

\begin{proof}[Proof of \cref{lem:up_pricing}]
Let the actual gap be $\Delta = p^{(j+1)} - p^{(j)}$. We know $0 < \Delta \le \epsilon$. We require $u_i(p^{(j)} + \Delta, p_{-i}) > u_i(p^{(j)}, p_{-i})$ for all $p_{-i} \ge p^{(j)}$.

Applying the exact cross-multiplication and $Z$ substitution steps as in \cref{lem:undercutting}, the strict dominance requirement simplifies to:
\[
\left( e^{\mu \cdot (p^{(j)}-p_{-i})} + \alpha e^{\mu p^{(j)}} \right) \left[ (p^{(j)} - c) (e^{\mu\Delta} - 1) - \Delta \right] < \Delta
\]
Here, the strict dominance requirement is strictly bounded above by $\Delta$. Because the left-hand side is strictly decreasing with respect to $p_{-i}$, its supremum over the opponent's strategy space occurs at the lower boundary, $p_{-i} = p^{(j)}$.
Substituting $p_{-i} = p^{(j)}$ and rearranging to isolate $p^{(j)}$ yields the condition:
\[
p^{(j)} < c + \frac{\Delta}{e^{\mu\Delta} - 1} \left[ 1 + \frac{1}{1 + \alpha e^{\mu p^{(j)}}} \right]
\]
Notice that the function $h(x) = \frac{x}{e^{\mu x} - 1}$ is strictly decreasing for $x > 0$. Because $\Delta \le \epsilon$, it follows that $h(\Delta) \ge h(\epsilon)$. Therefore, if $p^{(j)}$ falls strictly below the boundary evaluated at $\epsilon$, it falls below the boundary evaluated at $\Delta$.
Replacing $\Delta$ with $\epsilon$ recovers our lower threshold definition:
\[
p^{(j)} < \Phi^-(p^{(j)})
\]
By the non-increasing property of $\Phi^-(p)$, our assumption that $p^{(j)} < p^-$ ensures $p^{(j)} < p^- = \Phi^-(p^-) \le \Phi^-(p^{(j)})$. Thus, the strict inequality holds, establishing that the pure strategy $p^{(j)}$ is strictly dominated by $p^{(j+1)}$.
\end{proof}

Having established the bounds on the surviving strategy set for the Bertrand Logit grid, we now combine this structural result with our behavioral convergence guarantee (\cref{thm:da_IESPDS}). This yields our final corollary: when Domination-Avoiding agents interact in this discrete market, they asymptotically concentrate their pricing entirely within the competitive window.

\begin{customcorollary}{3}
    Consider two Domination-Avoiding agents that repeatedly play a symmetric Bertrand duopoly game with Logit demand over any finite grid of possible prices $p^{(1)}<p^{(2)}< \dots <p^{(M)}$, that is $\epsilon$-dense, $p^{(j)}-p^{(j-1)} \le \epsilon$, and whose range contains the equilibrium,  $p^{(1)} \le p^* \le p^{(M)}$ where $p^*$ is the competitive duopoly equilibrium price (of the continuous Bertrand game). Then
    \[
    \lim_{T \rightarrow \infty} \Pr_{t \sim {U}(1 \dots T)}[ p^*-2\epsilon \le p_i^t \le p^*+2\epsilon] = 1
    \] 
    where $p_i^t$ is the price played by player $i$ at time $t$ and where the probability is taken over a uniformly random time $t$ in the first $T$ steps as well as over the internal randomization of the agents.
\end{customcorollary}

\begin{proof}
Let $\mathcal{P}$ denote the finite price grid. By \cref{thm:IESPDS_collapse}, applying Iterated Elimination of Strictly Purely Dominated Strategies (IESPDS) to this game eliminates all prices strictly outside the neighborhood of the continuous equilibrium. Thus, the set of surviving prices satisfies $S^\infty \subseteq \{p \in \mathcal{P} : |p - p^*| \le 2\epsilon\}$.

Let $E = \mathcal{P} \setminus S^\infty$ denote the set of eliminated prices. By \cref{thm:da_IESPDS}, because the agents are Domination-Avoiding, the probability of playing any specific eliminated price $p \in E$ asymptotically converges to zero:
\[
\forall p \in E, \quad \lim_{T \rightarrow \infty} \Pr_{t \sim {U}(1 \dots T)} [p_i^t = p] = 0
\]
Because the discrete grid $\mathcal{P}$ is finite, the eliminated set $E$ is also strictly finite. Because the agents play exactly one price at any given time step, these events are mutually exclusive. Therefore, the asymptotic probability of playing \emph{any} eliminated price evaluates exactly to zero.
\[
\lim_{T \rightarrow \infty} \Pr_{t \sim {U}(1 \dots T)} [p_i^t \in E] = \sum_{p \in E} \lim_{T \rightarrow \infty} \Pr_{t \sim {U}(1 \dots T)} [p_i^t = p] = 0
\]
By the law of total probability, the agent must play a price from the full grid $\mathcal{P}$. Therefore, the probability mass must asymptotically concentrate entirely on the surviving set $S^\infty$:
\[
\lim_{T \rightarrow \infty} \Pr_{t \sim {U}(1 \dots T)} [p_i^t \in S^\infty] = 1 - \lim_{T \rightarrow \infty} \Pr_{t \sim {U}(1 \dots T)} [p_i^t \in E] = 1
\]
Because every surviving price satisfies $|p - p^*| \le 2\epsilon$, substitution yields:
\[
\lim_{T \rightarrow \infty} \Pr_{t \sim {U}(1 \dots T)}[ p^*-2\epsilon \le p_i^t \le p^*+2\epsilon] = 1
\]
completing the proof.
\end{proof}


\end{document}